%% file: byzantine-eventual.tex
\documentclass[sigconf,nonacm]{acmart}
\usepackage{algorithm}
\usepackage{algpseudocode} % for pseudocode
\usepackage{tikz} % for figures
\usepackage{subcaption} % for subfigures

\newcommand{\I}{$\mathcal{I}\!$}

\begin{document}
\title{Byzantine Eventual Consistency and the Fundamental Limits of Peer-to-Peer Databases}

\author{Martin Kleppmann}
\orcid{0000-0001-7252-6958}
\affiliation{%
  \institution{University of Cambridge}
  \city{Cambridge}
  \country{UK}
}
\email{mk428@cst.cam.ac.uk}

\author{Heidi Howard}
\orcid{0000-0001-5256-7664}
\affiliation{%
  \institution{University of Cambridge}
  \city{Cambridge}
  \country{UK}
}
\email{hh360@cst.cam.ac.uk}

\begin{abstract}
Sybil attacks, in which a large number of adversary-controlled nodes join a network, are a concern for many peer-to-peer database systems, necessitating expensive countermeasures such as proof-of-work.
However, there is a category of database applications that are, by design, immune to Sybil attacks because they can tolerate arbitrary numbers of Byzantine-faulty nodes.
In this paper, we characterize this category of applications using a consistency model we call \emph{Byzantine Eventual Consistency} (BEC).
We introduce an algorithm that guarantees BEC based on Byzantine causal broadcast, prove its correctness, and demonstrate near-optimal performance in a prototype implementation.
%We also show that \I-confluence (invariant confluence) is a necessary and sufficient condition for the existence of such an algorithm.
\end{abstract}
\maketitle
\pagestyle{plain}

\section{Introduction}

Peer-to-peer systems are of interest to many communities for a number of reasons: their lack of central control by a single party can make them more resilient, and less susceptible to censorship and denial-of-service attacks than centralized services.
Examples of widely deployed peer-to-peer applications include file sharing~\cite{Pouwelse:2005}, scientific dataset sharing~\cite{Robinson:2018}, decentralized social networking~\cite{Tarr:2019}, cryptocurrencies~\cite{Nakamoto:2008}, and blockchains~\cite{Bano:2019}.

Many peer-to-peer systems are essentially replicated database systems, albeit often with an application-specific data model.
For example, in a cryptocurrency, the replicated state comprises the balance of each user's account; in BitTorrent~\cite{Pouwelse:2005}, it is the files being shared.
Some blockchains support more general data storage and smart contracts (essentially, deterministic stored procedures) that are executed as serializable transactions by a consensus algorithm.

The central challenge faced by peer-to-peer databases is that peers cannot be trusted because anybody in the world can add peers to the network.
Thus, we must assume that some subset of peers are malicious; such peers are also called \emph{Byzantine-faulty}, which means that they may deviate from the specified protocol in arbitrary ways.
Moreover, a malicious party may perform a \emph{Sybil attack}~\cite{Douceur:2002}: launching a large number of peers, potentially causing the Byzantine-faulty peers to outnumber the honest ones.

Several countermeasures against Sybil attacks are used.
Bitcoin popularized the concept of \emph{proof-of-work}~\cite{Nakamoto:2008}, in which a peer's voting power depends on the computational effort it expends.
Unfortunately, proof-of-work is extraordinarily expensive: it has been estimated that as of 2020, Bitcoin alone represents almost half of worldwide datacenter electricity use~\cite{deVries:2020}.
\emph{Permissioned} blockchains avoid this huge carbon footprint, but they have the downside of requiring central control over the peers that may join the system, undermining the principle of decentralization.
Other mechanisms, such as \emph{proof-of-stake}~\cite{Bano:2019}, are at present unproven.

The reason why permissioned blockchains must control membership is that they rely on Byzantine agreement, which assumes that at most $f$ nodes are Byzantine-faulty.
To tolerate $f$ faults, Byzantine agreement algorithms typically require at least $3f+1$ nodes~\cite{Castro:1999}.
% It is well established that without synchrony, Byzantine agreement is impossible if $n<3f+1$~\cite{Dwork:1988,Lamport:1982}.
If more than $f$ nodes are faulty, these algorithms can guarantee neither safety (agreement) nor liveness (progress).
Thus, a Sybil attack that causes the bound of $f$ faulty nodes to be exceeded can result in the system's guarantees being violated; for example, in a cryptocurrency, they could allow the same coin to be spent multiple times (a \emph{double-spending} attack).

This state of affairs raises the question: if Byzantine agreement cannot be achieved in the face of arbitrary numbers of Byzantine-faulty nodes, what properties \emph{can} be guaranteed in this case?

A system that tolerates arbitrary numbers of Byzantine-faulty nodes is immune to Sybil attacks: even if the malicious peers outnumber the honest ones, it is still able to function correctly.
This makes such systems of large practical importance: being immune to Sybil attacks means neither proof-of-work nor the central control of permissioned blockchains is required.

In this paper, we provide a precise characterization of the types of problems that can and cannot be solved in the face of arbitrary numbers of Byzantine-faulty nodes.
We do this by viewing peer-to-peer networks through the lens of distributed database systems and their consistency models.
Our analysis is based on using \emph{invariants}~-- predicates over database states~-- to express an application's correctness properties, such as integrity constraints.

Our key result is a theorem stating that it is possible for a peer-to-peer database to be immune to Sybil attacks if and only if all of the possible transactions are \emph{\I-confluent} (invariant confluent) with respect to all of the application's invariants on the database.
\I-confluence, defined in Section~\ref{sec:confluence}, was originally introduced for non-Byzantine systems~\cite{Bailis:2014}, and our result shows that it is also applicable in a Byzantine context.
Our result does not solve the problem of Bitcoin electricity consumption, because (as we show later) a cryptocurrency is not \I-confluent.
However, there is a wide range of applications that \emph{are} \I-confluent, and which can therefore be implemented in a permissionless peer-to-peer system without resorting to proof-of-work.
Our work shows how to do this.

Our contributions in this paper are as follows:
\begin{enumerate}
    \item We define a consistency model for replicated databases, called \emph{Byzantine eventual consistency} (BEC), which can be achieved in systems with arbitrary numbers of Byzantine-faulty nodes.
    \item We introduce replication algorithms that ensure BEC, and prove their correctness without bounding the number of Byzantine faults.
    Our approach first defines \emph{Byzantine causal broadcast}, a mechanism for reliably multicasting messages to a group of nodes, and then uses it for BEC replication.
    \item We evaluate the performance of a prototype implementation of our algorithms, and demonstrate that our optimized algorithm incurs only a small network communication overhead, making it viable for use in practical systems.
    \item We prove that \I-confluence is a necessary and sufficient condition for the existence of a BEC replication algorithm, and we use this result to determine which applications can be immune to Sybil attacks.
\end{enumerate}

\section{Background and Definitions}

We first introduce background required for the rest of the paper.

\subsection{Strong Eventual Consistency and CRDTs}\label{sec:crdts}

Eventual consistency is usually defined as: ``If no further updates are made, then eventually all replicas will be in the same state~\cite{Vogels:2009ca}.''
This is a very weak model: it does not specify when the consistent state will be reached, and the premise ``if no further updates are made'' may never be true in a system in which updates happen continuously.
To strengthen this model, Shapiro et al.~\cite{Shapiro:2011} introduce \emph{strong eventual consistency} (SEC), which requires that:

\begin{description}
\item[Eventual update:] If an update is applied by a correct replica, then all correct replicas will eventually apply that update.
\item[Convergence:] Any two correct replicas that have applied the same set of updates are in the same state (even if the updates were applied in a different order).
\end{description}

Read operations can be performed on any replica at any time, and they return that replica's current state at that point in time.

In the context of replicated databases, one way of achieving SEC is by executing a transaction at one replica, disseminating the updates from the transaction to the other replicas using a reliable broadcast protocol (e.g.\ a gossip protocol~\cite{Leitao:2009fi}), and applying the updates to each replica's state using a commutative function.
Let $S' = \mathrm{apply}(S, u)$ be the function that applies the set of updates $u$ to the replica state $S$, resulting in an updated replica state $S'$.
Then two sets of updates $u_1$ and $u_2$ commute if
\[ \forall S.\; \mathrm{apply}(\mathrm{apply}(S, u_1), u_2) = \mathrm{apply}(\mathrm{apply}(S, u_2), u_1). \]
Two replicas can apply the same commutative sets of updates in a different order, and still converge to the same state.

One technique for implementing such commutativity is to use \emph{Conflict-free Replicated Data Types} (\emph{CRDTs})~\cite{Shapiro:2011}.
These abstract datatypes are designed such that concurrent updates to their state commute, with built-in resolution policies for conflicting updates.
CRDTs have been used to implement a range of applications, such as key-value stores~\cite{Akkoorath2016Cure,Zawirski2015SwiftCloud}, multi-user collaborative text editors~\cite{Weiss:2009ht}, note-taking tools~\cite{vanHardenberg2020PushPin}, games~\cite{vanderLinde:2017fu}, CAD applications~\cite{Lv:2018ie}, distributed filesystems~\cite{Najafzadeh:2018bw,Tao:2015gd}, project management tools~\cite{Kleppmann2019localfirst}, and many others.
Several papers present techniques for achieving commutativity in different datatypes~\cite{Kleppmann:2017,Preguica:2018gi,Shapiro:2011wy,Weiss:2009ht}.

\subsection{Invariant confluence}\label{sec:confluence}

An \emph{invariant} is a predicate over replica states, i.e.\ a function $I(S)$ that takes a replica state $S$ and returns either $\mathsf{true}$ or $\mathsf{false}$.
Invariants can represent many types of consistency properties and constraints commonly found in databases, such as referential integrity, uniqueness, or restrictions on the value of a data item (e.g.\ requiring it to be non-negative).

Informally, a set of transactions is \I-confluent with regard to an invariant $I$ if different replicas can independently execute subsets of the transactions, each ensuring that $I$ is preserved, and we can be sure that the result of merging the updates from those transactions will still satisfy $I$.
More formally, let $T = \{T_1, \dots, T_n\}$ be the set of transactions executed by a system, and let $u_i$ be the updates resulting from the execution of $T_i$.
Assume that for all $i,j \in [1,n]$, if $T_i$ and $T_j$ were executed concurrently by different replicas (written $T_i \parallel T_j$), then updates $u_i$ and $u_j$ commute.
Then we say that $T$ is \I-confluent with regard to invariant $I$ if:
\begin{align*}
    \forall i,j \in [1,n].\; \forall S.\; & (T_i \parallel T_j) \wedge I(S) \wedge I(\mathrm{apply}(S, u_i)) \wedge I(\mathrm{apply}(S, u_j)) \\
    & \Longrightarrow I(\mathrm{apply}(\mathrm{apply}(S, u_i), u_j)).
\end{align*}
As $u_i$ and $u_j$ commute, this also implies $I(\mathrm{apply}(\mathrm{apply}(S, u_j), u_i))$.

As an example, consider a uniqueness constraint, i.e.\ $I(S)=\mathsf{true}$ if there is no more than one data item in $S$ for which a particular attribute has a given value.
If $T_1$ and $T_2$ are both transactions that create data items with the same value in that attribute, then $\{T_1,T_2\}$ is not \I-confluent with regard to $I$: each of $T_1$ and $T_2$ individually preserves the constraint, but the combination of the two does not.

As a second example, say that $I(S)=\mathsf{true}$ if every user in $S$ has a non-negative account balance.
If $T_1$ and $T_2$ are both transactions that increase the same user's account balance, then $\{T_1,T_2\}$ is \I-confluent with regard to $I$, because the sum of the two positive balance updates cannot cause the balance to become negative (assuming no overflow).
However, if $T_1$ and $T_2$ decrease the same user's account balance, then they are not \I-confluent with regard to $I$: any one of the transactions may be fine, but the sum of the two could cause the balance to become negative.

\I-confluence was introduced by Bailis et al.~\cite{Bailis:2014} in the context of non-Byzantine systems, along with a proof that a set of transactions can be executed in a \emph{coordination-free} manner if and only if those transactions are \I-confluent with regard to all of the application's invariants.
``Coordination-free'' means, loosely speaking, that one replica does not have to wait for a response from any other replica before it can commit a transaction.

\subsection{System model}\label{sec:system-model}

Our system consists of a finite set of replicas, which may vary over time.
Any replica may execute transactions.
Each replica is either \emph{correct} or \emph{faulty}, but a correct replica does not know whether another replica is faulty.
A correct replica follows the specified protocol, whereas a faulty replica may deviate from the protocol in arbitrary ways (i.e.\ it is Byzantine-faulty~\cite{Lamport:1982}).
Faulty replicas may collude and attempt to deceive correct replicas; we model such worst-case behavior by assuming a malicious adversary who controls the behavior of all faulty replicas.
We allow any subset of replicas to be faulty.
We consider all replicas to be equal peers, making no distinction e.g.\ between clients and servers.
Replicas may crash and recover; as long as a crashed replica eventually recovers, and otherwise follows the protocol, we still call it ``correct''.

We assume that each replica has a distinct private key that can be used for digital signatures, and that the corresponding public key is known to all replicas.
We assume that no replica knows the private key of another replica, and thus signatures cannot be forged.
Unlike in a permissioned blockchain, there is no need for central control over the set of public keys in the system: for example, one replica may add another replica to the system by informing the existing replicas about the new replica's public key.

Replicas communicate by sending messages over pairwise (bidirectional, unicast) network links.
We assume that all messages sent over these links are signed with the sender's private key, and the recipient ignores messages with invalid signatures.
Thus, even if the adversary can tamper with network traffic, it can only cause message loss but not impersonate a correct replica.
For simplicity, our algorithms assume that network links are reliable, i.e.\ that a sent message is eventually delivered, provided that neither sender nor recipient crashes.
This can easily be achieved by detecting and retransmitting any lost messages.

We make no timing assumptions: messages may experience unbounded delay in the network (for example, due to retransmissions during temporary network partitions), replicas may execute at different speeds, and we do not assume any clock synchronization (i.e.\ we assume an \emph{asynchronous} system model).
We do assume that a replica has a timer that allows it to perform some task approximately periodically, such as retransmitting unacknowledged messages, without requiring exact time measurement.

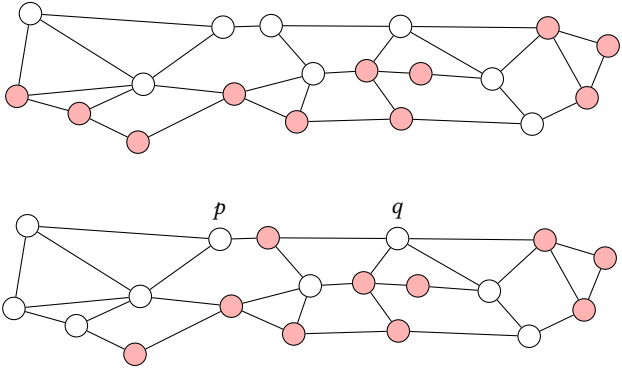
\begin{figure}
    \centering
    \input{figs/connected.tikz}
    \caption{Above: correct replicas (white) form a connected component. Below: faulty replicas (red) are able to prevent communication between correct replicas $p$ and $q$.}
    \label{fig:connected}
\end{figure}

Not all pairs of replicas are necessarily connected with a network link.
However, we must assume that in the graph of replicas and network links, the correct replicas form a single connected component, as illustrated in Figure~\ref{fig:connected}.
This assumption is necessary because if two correct replicas can only communicate via faulty replicas, then no algorithm can guarantee data exchange between those replicas, as the adversary can always block communication (this is known as an \emph{eclipse attack}~\cite{Singh:2004}).
The easiest way of satisfying this assumption is to connect each replica to every other.

\section{Byzantine Eventual Consistency}\label{sec:byzantine-crdts}

We now define Byzantine Eventual Consistency (BEC), and prove that \I-confluence is both necessary and sufficient to implement it.

\subsection{Definition of BEC}\label{sec:bec-definition}

We say that a replica \emph{generates} a set of updates if those updates are the result of that replica executing a committed transaction.
We say a replicated database provides \emph{Byzantine Eventual Consistency} if it satisfies the following properties in the system model of \S~\ref{sec:system-model}:

\begin{description}
\item[Self-update:] If a correct replica generates an update, it applies that update to its own state.
\item[Eventual update:] For any update applied by a correct replica, all correct replicas will eventually apply that update.
\item[Convergence:] Any two correct replicas that have applied the same set of updates are in the same state.
\item[Atomicity:] When a correct replica applies an update, it atomically applies all of the updates resulting from the same transaction.
\item[Authenticity:] If a correct replica applies an update that is labeled as originating from replica $s$, then that update was generated by replica $s$.
%\item[Non-duplication:] A correct replica does not apply the same update more than once.
\item[Causal consistency:] If a correct replica generates or applies update $u_1$ before generating update $u_2$, then all correct replicas apply $u_1$ before $u_2$.
\item[Invariant preservation:] The state of a correct replica always satisfies all of the application's declared invariants.
\end{description}

Read operations can be performed at any time, and their result reflects the replica state that results from applying only updates made by committed transactions.
In other words, we require \emph{read committed} transaction isolation~\cite{Adya:2000}, but we do not assume serializable isolation.
BEC is a strengthening of SEC (\S~\ref{sec:crdts}); the main differences are that SEC assumes a non-Byzantine system, and SEC does not require atomicity, causal consistency, or invariant preservation.

BEC ensures that all correct replicas converge towards the same shared state, even if they also communicate with any number of Byzantine-faulty replicas.
Essentially, BEC ensures that faulty replicas cannot permanently corrupt the state of correct replicas.
As is standard in Byzantine systems, the properties above only constrain the behavior of correct replicas, since we can make no assumptions or guarantees about the behavior or state of faulty replicas.

\subsection{Existence of a BEC algorithm}\label{sec:theorem}

For the following theorem we define a replication algorithm to be \emph{fault-tolerant} if it is able to commit transactions while at most one replica is crashed or unreachable.
In other words, in a system with $r$ replicas, a replica is able to commit a transaction after receiving responses from up to $r-2$ replicas (all but itself and one unavailable replica).
We adopt this very weak definition of fault tolerance since it makes the following theorem stronger; the theorem also holds for algorithms that tolerate more than one fault.

We are now ready to prove our main theorem:
\begin{theorem}\label{theorem}
Assume an asynchronous\footnote{This theorem also holds for partially synchronous~\cite{Dwork:1988} systems, in which network latency is only temporarily unbounded, but eventually becomes bounded. However, for simplicity, we assume an asynchronous system in this proof.} system with a finite set of replicas, of which any subset may be Byzantine-faulty.
Assume there is a known set of invariants that the data on each replica should satisfy.
Then there exists a fault-tolerant algorithm that ensures BEC if and only if the set of all transactions executed by correct replicas is \I-confluent with respect to each of the invariants.
\end{theorem}

\begin{proof}
For the backward direction, we assume that the set of transactions executed by correct replicas is \I-confluent with respect to all of the invariants.
Then the algorithm defined in \S~\ref{sec:algorithm} ensures BEC, as proved in Appendices~\ref{sec:proof} and~\ref{sec:bec-proof}, demonstrating the existence of an algorithm that ensures BEC.

For the forward direction, we assume that the set of transactions $T$ executed by correct replicas is not \I-confluent with respect to at least one invariant $I$. We must then show that under this assumption, there is no fault-tolerant algorithm that ensures BEC and preserves all invariants in the presence of an arbitrary number of Byzantine-faulty replicas.
We do this by assuming that such an algorithm exists and deriving a contradiction.

If $T$ is not \I-confluent with respect to $I$, then there must exist concurrently executed transactions $T_i, T_j \in T$ that violate \I-confluence.
That is, $u_i$ and $u_j$ are the sets of updates generated by $T_i$ and $T_j$ respectively, and there exists a replica state $S$ such that
\[ I(S) \wedge I(\mathrm{apply}(S, u_i)) \wedge I(\mathrm{apply}(S, u_j)) \wedge \neg I(\mathrm{apply}(\mathrm{apply}(S, u_i), u_j)). \]

Let $R$ be the set of replicas.
Let $p$ be the correct replica that executes $T_i$, let $q$ be the correct replica that executes $T_j$, and assume that all of the remaining replicas $R \setminus \{p,q\}$ are Byzantine-faulty.
Assume $p$ and $q$ are both in the state $S$ before executing $T_i$ and $T_j$.

Now we let $p$ and $q$ execute $T_i$ and $T_j$ concurrently.
The transaction execution and replication algorithm may perform arbitrary computation and communication among replicas.
However, since the system is asynchronous, messages may be subject to unbounded network latency.
Assume that in this execution, messages between $p$ and $q$ are severely delayed, while messages between any other pairs of replicas are received quickly.

Since the replication algorithm is fault-tolerant, replica $p$ must eventually commit $T_i$ without receiving any message from $q$, and similarly $q$ must eventually commit $T_j$ without receiving any message from $p$.
Both transactions may communicate with any subset of $R \setminus \{p,q\}$, but since all of these replicas are Byzantine-faulty, they may fail to inform $p$ about $q$'s conflicting transaction $T_j$, and fail to inform $q$ about $p$'s conflicting transaction $T_i$.
Thus, $T_i$ and $T_j$ are both eventually committed.

After both $T_i$ and $T_j$ have been committed, communication between $p$ and $q$ becomes fast again.
Due to the \emph{eventual update} property of BEC, $u_i$ must eventually be applied at $q$, and $u_j$ must eventually be applied at $p$, resulting in the state $\mathrm{apply}(\mathrm{apply}(S, u_i), u_j)$ on both replicas, in which $I$ is violated.
This contradicts our earlier assumption that the algorithm always preserves invariants.

Since we did not make any assumptions about the internal structure of the algorithm, this argument shows that no fault-tolerant algorithm exists that guarantees BEC in this setting.
\end{proof}

\subsection{Discussion}\label{sec:bec-discussion}

Theorem~\ref{theorem} shows us that an application can be implemented in a system with arbitrarily many Byzantine-faulty replicas if and only if its transactions are \I-confluent with respect to its invariants.
It is both a negative (impossibility) and a positive (existence) result.

As an example of impossibility, consider a cryptocurrency, which must reduce a user's account balance when a user makes a payment, and which must ensure that a user does not spend more money than they have.
As we saw in \S~\ref{sec:confluence}, payment transactions from the same payer are not \I-confluent with regard to the account balance invariant, and thus a cryptocurrency cannot be immune to Sybil attacks.
For this reason, it needs Sybil countermeasures such as proof-of-work or centrally managed permissions.

On the other hand, many of the CRDT applications listed in \S~\ref{sec:crdts} only require \I-confluent transactions and invariants.
Theorem~\ref{theorem} shows that it is possible to implement such applications without any Sybil countermeasures, because it is possible to ensure BEC and preserve all invariants regardless of how many Byzantine-faulty replicas are in the system.

Even in applications that are not fully \I-confluent, our result shows that the \I-confluent portions of the application can be implemented without incurring the costs of Sybil countermeasures, and a Byzantine consensus algorithm need only be used for those transactions that are not \I-confluent with respect to the application's invariants.
For example, an auction could aggregate bids in an \I-confluent manner, and only require consensus to decide the winning bid.
As another example, most of the transactions and invariants in the TPC-C benchmark are \I-confluent~\cite{Bailis:2014}.
% TODO: citation for the auction example? I saw it in a talk somewhere, can't remember which...

\section{Background on broadcast}\label{sec:broadcast}

Before we introduce our algorithms for ensuring BEC in \S~\ref{sec:algorithm}, we first give some additional background and highlight some of the difficulties of working in a Byzantine system model.

\subsection{Reliable, causal, and total order broadcast}\label{sec:broadcast-properties}

We implement BEC replication by first defining a \emph{broadcast} protocol, and then layering replication on top of it.
Several different forms of broadcast have been defined in the literature~\cite{Cachin:2011wt}, and we now introduce them briefly.
Broadcast protocols are defined in terms of two primitives, \emph{broadcast} and \emph{deliver}.
Any replica (or node) in the system may broadcast a message, and we want all replicas to deliver messages that were broadcast.

\emph{Reliable broadcast} must satisfy the following properties:

\begin{description}
\item[Self-delivery:] If a correct replica $p$ broadcasts a message $m$, then $p$ eventually delivers $m$.
\item[Eventual delivery:] If a correct replica delivers a message $m$, then all correct replicas will eventually deliver $m$.
\item[Authenticity:] If a correct replica delivers a message $m$ with sender $s$, then $m$ was broadcast by $s$.
\item[Non-duplication:] A correct replica does not deliver the same message more than once.
\end{description}

Reliable broadcast does not constrain the order in which messages may be delivered.
In many applications the delivery order is important, so we can strengthen the model.
For example, \emph{total order broadcast} must satisfy the four properties of reliable broadcast, and additionally the following property:

\begin{description}
\item[Total order:] If a correct replica delivers message $m_1$ before delivering message $m_2$, then all correct replicas must deliver $m_1$ before delivering $m_2$.
\end{description}

Total order broadcast ensures that all replicas deliver the same messages in the same order~\cite{Defago:2004ji}.
It is a very powerful model, since it can for example implement serializable transactions (by encoding each transaction as a stored procedure in a message, and executing them in the order they are delivered at each replica) and \emph{state machine replication}~\cite{Schneider:1990} (providing linearizable replicated storage).

In a Byzantine system, total order broadcast is implemented by Byzantine agreement algorithms.
An example is a blockchain, in which the totally ordered chain of blocks corresponds to the sequence of delivered messages~\cite{Bano:2019}.
However, Byzantine agreement algorithms must assume a maximum number of faulty replicas (see \S~\ref{sec:relwork}), and hence require Sybil countermeasures.
To ensure eventual delivery they must also assume partial synchrony~\cite{Dwork:1988}.

%State machine replication treats every replica as a deterministic state machine, where the inputs are commands.
%If all replicas observe the same commands in the same order, they all go through the same sequence of state transitions, resulting in the same final state.

\emph{Causal broadcast}~\cite{Birman:1991el,Cachin:2011wt} must satisfy the four properties of reliable broadcast, and additionally the following ordering property:

\begin{description}
\item[Causal order:] If a correct replica broadcasts or delivers $m_1$ before broadcasting message $m_2$, then all correct replicas must deliver $m_1$ before delivering $m_2$.
\end{description}

Causal order is based on the observation that when a replica broadcasts a message, that message may depend on prior messages seen by that replica (these are \emph{causal dependencies}).
It then imposes a partial order on messages: $m_1$ must be delivered before $m_2$ if $m_2$ has a causal dependency on $m_1$.
Concurrently sent messages, which do not depend on each other, can be delivered in any order.

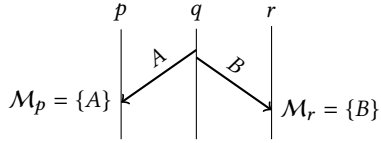
\begin{figure}
    \centering
    \input{figs/trivial1.tikz}
    \captionsetup{width=.95\linewidth}
    \caption{Byzantine-faulty replica $q$ sends conflicting messages to correct replicas $p$ and $r$.
    The sets $\mathcal{M}_p$ and ${M}_r$ do not converge.}
    \label{fig:trivial1}
\end{figure}

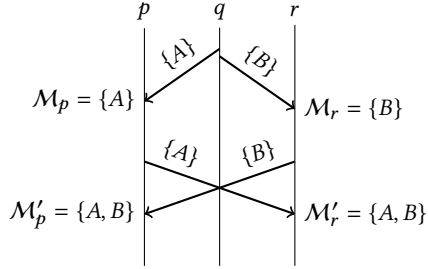
\begin{figure}
    \centering
    \input{figs/trivial2.tikz}
    \captionsetup{width=.95\linewidth}
    \caption{As correct replicas $p$ and $r$ reconcile their sets of messages, they converge to the same set $\mathcal{M}_p' = \mathcal{M}_r' = \{A,B\}$.}
    \label{fig:trivial2}
\end{figure}

\subsection{Na\"{\i}ve broadcast algorithms}\label{sec:naive-broadcast-algorithms}

The simplest broadcast algorithm is as follows: every time a replica wants to broadcast a message, it delivers that message to itself, and also sends that message to each other replica via a pairwise network link, re-transmitting until it is acknowledged.
However, this algorithm does not provide the \emph{eventual delivery} property in the face of Byzantine-faulty replicas, as shown in Figure~\ref{fig:trivial1}: a faulty replica $q$ may send two different messages $A$ and $B$ to correct replicas $p$ and $r$, respectively; then $p$ never delivers $B$ and $r$ never delivers $A$.

To address this issue, replicas $p$ and $r$ must communicate with each other (either directly, or indirectly via other correct replicas).
Let $\mathcal{M}_p$ and $\mathcal{M}_r$ be the set of messages delivered by replicas $p$ and $r$, respectively.
Then, as shown in Figure~\ref{fig:trivial2}, $p$ can send its entire set $\mathcal{M}_p$ to $r$, and $r$ can send $\mathcal{M}_r$ to $p$, so that both replicas can compute $\mathcal{M}_p \cup \mathcal{M}_r$, and deliver any new messages.
Pairs of replicas can thus periodically \emph{reconcile} their sets of delivered messages.

Adding this reconciliation process to the protocol ensures reliable broadcast.
However, this algorithm is very inefficient: when replicas periodically reconcile their state, we can expect that at the start of each round of reconciliation their sets of messages already have many elements in common.
Sending the entire set of messages to each other transmits a large amount of data unnecessarily.

An efficient reconciliation algorithm should determine which messages have already been delivered by both replicas, and transmit only those messages that are unknown to the other replica.
For example, replica $p$ should only send $\mathcal{M}_p \setminus \mathcal{M}_r$ to replica $r$, and replica $r$ should only send $\mathcal{M}_r \setminus \mathcal{M}_p$ to replica $p$.
The algorithm should also complete in a small number of round-trips and minimize the size of messages sent.
These goals rule out other na\"{\i}ve approaches too: for example, instead of sending all messages in $\mathcal{M}_p$, replica $p$ could send the hash of each message in $\mathcal{M}_p$, which can be used by other replicas to determine which messages they are missing; this is still inefficient, as the message size is $O(|\mathcal{M}_{p}|)$.

\begin{figure*}
    \centering
    \input{figs/vectorclocks.tikz}
    \caption{Replicas $p$ and $r$ believe they are in the same state because their vector timestamps are the same, when in fact their sets of messages are inconsistent due to $q$'s faulty behavior.}
    \label{fig:vectorclocks}
\end{figure*}
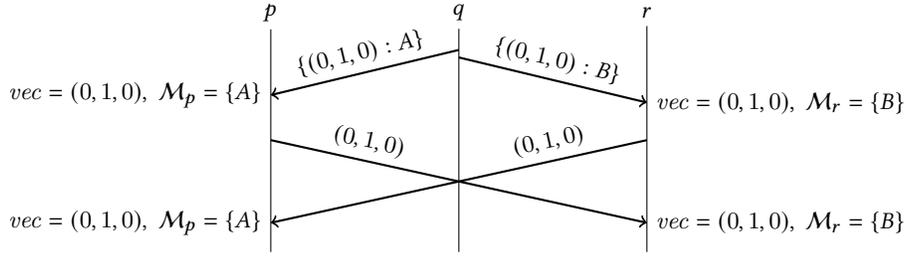

\subsection{Vector clocks}\label{sec:vectorclocks}

Non-Byzantine causal broadcast algorithms often rely on \emph{vector clocks} to determine which messages to send to each other, and how to order them~\cite{Birman:1991el,Schwarz:1994}.
However, vector clocks are not suitable in a Byzantine setting.
The problem is illustrated in Figure~\ref{fig:vectorclocks}, where faulty replica $q$ generates two different messages, $A$ and $B$, with the same vector timestamp $(0, 1, 0)$.

In a non-Byzantine system, the three components of the timestamp represent the number of distinct messages seen from $p$, $q$, and $r$ respectively.
Thus, $p$ and $r$ should be able to reconcile their sets of messages by first sending each other their latest vector timestamps as a concise summary of the set of messages they have seen.
However, in Figure~\ref{fig:vectorclocks} this approach fails due to $q$'s earlier faulty behavior: $p$ and $r$ detect that their vector timestamps are equal, and thus incorrectly believe that they are in the same state, even though their sets of messages are different.
Thus, vector clocks can be corrupted by a faulty replica.
A causal broadcast algorithm in a Byzantine system must not be vulnerable to such corruption.

\section{Algorithms for BEC}\label{sec:algorithm}

We now demonstrate how to implement BEC and therefore preserve \I-confluent invariants in a system with arbitrarily many Byzantine faults.
We begin by first presenting two causal broadcast algorithms (\S~\ref{sec:algorithm1} and \S~\ref{sec:algorithm2}), and then defining a replication algorithm on top (\S~\ref{sec:implementing-bec}).
At the core of our protocol is a reconciliation algorithm that ensures two replicas have delivered the same set of broadcast messages, in causal order.
The reconciliation is efficient in the sense that when two correct replicas communicate, they only exchange broadcast messages that the other replica has not already delivered.

\subsection{Definitions}\label{sec:algorithm-definitions}

Let $\mathcal{M}$ be the set of broadcast messages delivered by some replica.
$\mathcal{M}$ is a set of triples $(v, \mathit{hs}, \mathit{sig})$, where $v$ is any value, $\mathit{sig}$ is a digital signature over $(v, \mathit{hs})$ using the sender's private key, and $\mathit{hs}$ is a set of hashes produced by a cryptographic hash function $H(\cdot)$.
We assume that $H$ is collision-resistant, i.e.\ that it is computationally infeasible to find distinct $x$ and $y$ such that $H(x) = H(y)$.
This assumption is standard in cryptography, and it can easily be met by using a strong hash function such as SHA-256~\cite{SHA2}.

Let $A, B \in \mathcal{M}$, where $B = (v, \mathit{hs}, \mathit{sig})$ and $H(A) \in \mathit{hs}$.
Then we call $A$ a \emph{predecessor} of $B$, and $B$ a \emph{successor} of $A$.
Predecessors are also known as \emph{causal dependencies}.

Define a graph with a vertex for each message in $\mathcal{M}$, and a directed edge from each message to each of its predecessors.
We can assume that this graph is acyclic because the presence of a cycle would imply knowledge of a collision in the hash function.
Figure~\ref{fig:example-dags} shows examples of such graphs.

Let $\mathrm{succ}^1(\mathcal{M}, m)$ be the set of successors of message $m$ in $\mathcal{M}$, let $\mathrm{succ}^2(\mathcal{M}, m)$ be the successors of the successors of $m$, and so on, and let $\mathrm{succ}^*(\mathcal{M}, m)$ be the transitive closure:
\begin{align*}
\mathrm{succ}^i(\mathcal{M}, m) &=
\begin{cases}
\{(v, \mathit{hs}, \mathit{sig}) \in \mathcal{M} \mid H(m) \in \mathit{hs}\} & \text{ for } i=1 \\
\bigcup_{m' \in \mathrm{succ}^1(\mathcal{M}, m)} \mathrm{succ}^{i-1}(\mathcal{M}, m') & \text{ for } i>1
\end{cases} \\
\mathrm{succ}^*(\mathcal{M}, m) &= \bigcup_{i \ge 1} \mathrm{succ}^i(\mathcal{M}, m)
\end{align*}
We define the set of predecessors of $m$ similarly:
\begin{align*}
\mathrm{pred}^i(\mathcal{M}, m) &=
\begin{cases}
\{ m' \in \mathcal{M} \mid m = (v, \mathit{hs}, \mathit{sig}) \wedge H(m') \in \mathit{hs}\} & \text{ for } i=1 \\
\bigcup_{m' \in \mathrm{pred}^1(\mathcal{M}, m)} \mathrm{pred}^{i-1}(\mathcal{M}, m') & \text{ for } i>1
\end{cases} \\
\mathrm{pred}^*(\mathcal{M}, m) &= \bigcup_{i \ge 1} \mathrm{pred}^i(\mathcal{M}, m)
\end{align*}
Let $\mathrm{heads}(\mathcal{M})$ denote the set of hashes of those messages in $\mathcal{M}$ that have no successors:
\[ \mathrm{heads}(\mathcal{M}) = \{H(m) \mid m \in \mathcal{M} \wedge \mathrm{succ}^1(\mathcal{M}, m) = \{\}\;\}. \]

\subsection{Algorithm for Byzantine Causal Broadcast}\label{sec:algorithm1}

Define a \emph{connection} to be a logical grouping of a bidirectional sequence of related request/response messages between two replicas (in practice, it can be implemented as a TCP connection).
Our reconciliation algorithm runs in the context of a connection.

When a correct replica wishes to broadcast a message with value $v$, it executes lines~\ref{line:broadcast-begin}--\ref{line:broadcast-end} of Algorithm~\ref{fig:algorithm}: it constructs a message $m$ containing the current heads and a signature, delivers $m$ to itself, adds $m$ to the set of locally delivered messages $\mathcal{M}$, and sends $m$ via all connections.
However, this is not sufficient to ensure eventual delivery, since some replicas may be disconnected, and faulty replicas might not correctly follow this protocol.

To ensure eventual delivery, we assume that replicas periodically attempt to reconnect to each other and reconcile their sets of messages to discover any missing messages.
If two replicas are not able to connect directly, they can still exchange messages by periodically reconciling with one or more correct intermediary replicas (as stated in \S~\ref{sec:system-model}, we assume that such intermediaries exist).

\algblockdefx{On}{EndOn}[1]{\textbf{on} #1 \textbf{do}}{\textbf{end on}}
\algblockdefx{Atomic}{EndAtomic}{\textbf{atomically do}}{\textbf{end atomic}}

\begin{algorithm}[p]
    \begin{algorithmic}[1]
    \On{request to broadcast $v$}\label{line:broadcast-begin}
        \State $\mathit{hs} := \mathrm{heads}(\mathcal{M})$\label{line:broadcast-heads}
        \State $\mathit{sig} := \text{signature over } (v, \mathit{hs}) \text{ using this replica's private key}$
        \State $m := (v, \mathit{hs}, \mathit{sig})$
        \Atomic
            \State \textbf{deliver} $m$ to self\label{line:deliver-local}
            \State $\mathcal{M} := \mathcal{M} \cup \{m\}$\label{line:update-m-local}
        \EndAtomic
        \State \textbf{send} $\langle\mathsf{msgs}: \{m\}\rangle$ via all active connections\label{line:eager-send}
    \EndOn\label{line:broadcast-end}
    \State
    \On{connecting to another replica, and periodically} \label{line:connect-begin}
        \State // connection-local variables
        \State $\mathit{sent} := \{\};\; \mathit{recvd} := \{\};\; \mathit{missing} := \{\};\; \mathcal{M}_\mathsf{conn} := \mathcal{M}$ \label{line:init}
        \State \textbf{send} $\langle\mathsf{heads}: \mathrm{heads}(\mathcal{M}_\mathsf{conn})\rangle$ via current connection \label{line:send-heads}
    \EndOn \label{line:connect-end}
    \State
    \On{receiving $\langle\mathsf{heads}: \mathit{hs}\rangle$ via a connection} \label{line:recv-heads}
        \State \Call{HandleMissing}{$\{h \in \mathit{hs} \mid \nexists m \in \mathcal{M}_\mathsf{conn}.\; H(m) = h\}$} \label{line:heads-missing}
    \EndOn\label{line:recv-heads-end}
    \State
    \On{receiving $\langle\mathsf{msgs}: \mathit{new}\rangle$ via a connection} \label{line:recv-msgs}
        \State \mbox{$\mathit{recvd} := \mathit{recvd} \,\cup\, \{(v, \mathit{hs}, \mathit{sig}) \in \mathit{new} \mid \mathrm{check}((v, \mathit{hs}), \mathit{sig}) \}$}\label{line:msgs-recvd}
        \State $\mathit{unresolved} := \{h \mid \exists (v, \mathit{hs}, \mathit{sig}) \in \mathit{recvd}.\; h \in \mathit{hs} \;\wedge$ \label{line:msgs-missing}
        \State \hspace*{7.5em}$\nexists m \in (\mathcal{M}_\mathsf{conn} \cup \mathit{recvd}).\; H(m) = h\}$
        \State \Call{HandleMissing}{$\mathit{unresolved}$} \label{line:msgs-handle-missing}
    \EndOn\label{line:recv-msgs-end}
    \State
    \On{receiving $\langle\mathsf{needs}: \mathit{hashes}\rangle$ via a connection} \label{line:recv-needs}
        \State $\mathit{reply} := \{m \in \mathcal{M}_\mathsf{conn} \mid H(m) \in \mathit{hashes} \,\wedge\, m \notin \mathit{sent}\}$ \label{line:needs-reply}
        \State $\mathit{sent} := \mathit{sent} \cup \mathit{reply}$
        \State \textbf{send} $\langle\mathsf{msgs}: \mathit{reply}\rangle$ via current connection \label{line:send-msgs}
    \EndOn\label{line:end-needs}
    \State
    \Function{HandleMissing}{$\mathit{hashes}$}
        \State $\mathit{missing} := (\mathit{missing} \cup \mathit{hashes}) \setminus \{H(m) \mid m \in \mathit{recvd}\}$
        \If{$\mathit{missing} = \{\}$} \label{line:missing-empty}
            \Atomic
            \State $\mathit{msgs} := \mathit{recvd} \setminus \mathcal{M}$
            \State $\mathcal{M} := \mathcal{M} \cup \mathit{recvd}$ \label{line:update-m}
            \State \textbf{deliver} all of the messages in $\mathit{msgs}$ \label{line:deliver}
            \State \hspace{3.5em}in topologically sorted order
            \EndAtomic
            \State \textbf{send} $\langle\mathsf{msgs}: \mathit{msgs}\rangle$ via all other connections\label{line:eager-relay}
            \State \textbf{reconciliation complete} \label{line:finish}
        \Else
            \State \textbf{send} $\langle\mathsf{needs}: \mathit{missing}\rangle$ via current connection \label{line:send-missing}
        \EndIf
    \EndFunction
    \end{algorithmic}
    \caption{A Byzantine causal broadcast algorithm.}\label{fig:algorithm}
\end{algorithm}

\begin{figure}[p]
    \begin{subfigure}{\columnwidth}
    \centering
    \input{figs/dag-before-p.tikz}
    \caption{Messages at $p$ before reconciliation.}
    \end{subfigure}\\[1em]
    \begin{subfigure}{\columnwidth}
    \centering
    \input{figs/dag-before-q.tikz}
    \caption{Messages at $q$ before reconciliation.}
    \end{subfigure}\\[1em]
    \begin{subfigure}{\columnwidth}
    \centering
    \input{figs/dag-after.tikz}
    \caption{Messages at $p$ and $q$ after reconciliation.}
    \end{subfigure}
    \caption{Example DAGs of delivered messages. Arrows represent a message referencing the hash of its predecessor, and heads (messages with no successors) are marked with circles.}
    \label{fig:example-dags}
\end{figure}
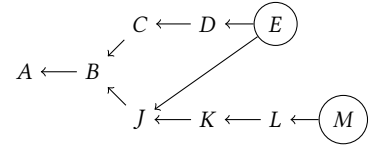
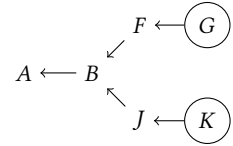
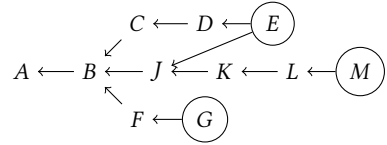

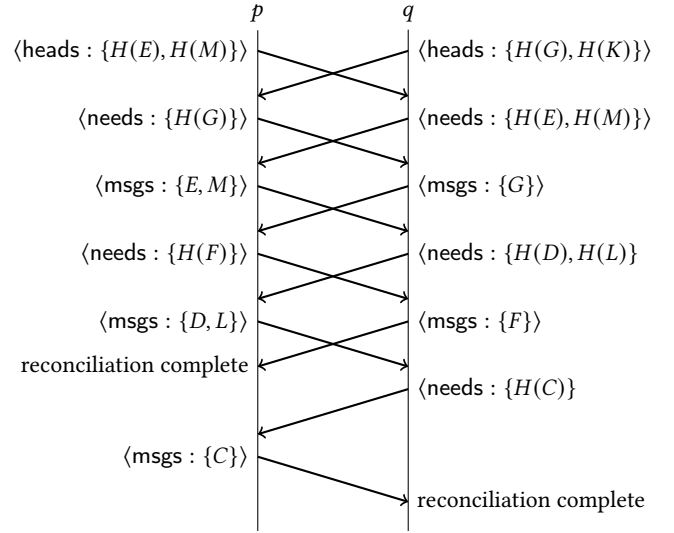
\begin{figure}[p]
    \input{figs/message-exchange.tikz}
    \caption{Requests/responses sent in the course of running the reconciliation process in Algorithm~\ref{fig:algorithm} with the example in Figure~\ref{fig:example-dags}.}
    \label{fig:messages}
\end{figure}

\begin{algorithm*}
    \begin{algorithmic}[1]
    \On{connecting to replica $q$, and periodically}\Comment{Replaces lines \ref{line:connect-begin}--\ref{line:connect-end} of Algorithm \ref{fig:algorithm}\hspace*{5cm}~}
        \State $\mathit{sent} := \{\};\; \mathit{recvd} := \{\};\; \mathit{missing} := \{\};\; \mathcal{M}_\mathsf{conn} := \mathcal{M}$ \Comment{connection-local variables\hspace*{5cm}~}
        \State $\mathit{oldHeads} := \Call{LoadHeads}{q}$\label{line:load-heads}
        \State $\mathit{filter} := \textsc{MakeBloomFilter}(\Call{MessagesSince}{\mathit{oldHeads}})$\label{line:make-bloom}
        \State \textbf{send} $\langle\mathsf{heads}: \mathrm{heads}(\mathcal{M}_\mathsf{conn}),\, \mathsf{oldHeads}: \mathit{oldHeads},\, \mathsf{filter}: \mathit{filter}\rangle$ \label{line:a2-send-heads}
    \EndOn
    \State
    \On{receiving $\langle\mathsf{heads}: \mathit{hs},\, \mathsf{oldHeads}: \mathit{oldHeads},\, \mathsf{filter}: \mathit{filter}\rangle$}\Comment{Replaces lines \ref{line:recv-heads}--\ref{line:recv-heads-end}\hspace*{5cm}~}\label{line:a2-recv-heads}
        \State $\mathit{bloomNegative} := \{m \in \Call{MessagesSince}{\mathit{oldHeads}} \mid \neg\Call{BloomMember}{\mathit{filter}, m}\}$\label{line:bloom-member}
        \State $\mathit{reply} := \left(\mathit{bloomNegative} \,\cup\, \bigcup_{m \in \mathit{bloomNegative}} \mathrm{succ}^*(\mathcal{M}_\mathsf{conn}, m)\right) \setminus \mathit{sent}$\label{line:bloom-succ}
        \If{$\mathit{reply} \neq \{\}$}
            \State $\mathit{sent} := \mathit{sent} \cup \mathit{reply}$
            \State \textbf{send} $\langle\mathsf{msgs}: \mathit{reply}\rangle$ \label{line:a2-heads-reply}
        \EndIf
        \State \Call{HandleMissing}{$\{h \in \mathit{hs} \mid \nexists m \in \mathcal{M}_\mathsf{conn}.\; H(m) = h\}$} \label{line:a2-heads-missing}
    \EndOn
    \State
    \Function{MessagesSince}{$\mathit{oldHeads}$}\label{line:msg-since-begin}
        \State $\mathit{known} := \{m \in \mathcal{M}_\mathsf{conn} \mid H(m) \in \mathit{oldHeads}\}$
        \State \textbf{return} $\mathcal{M}_\mathsf{conn} \setminus \left(\mathit{known} \,\cup\, \bigcup_{m \in \mathit{known}} \mathrm{pred}^*(\mathcal{M}_\mathsf{conn}, m)\right)$
    \EndFunction\label{line:msg-since-end}
    \end{algorithmic}
    \caption{Optimizing Algorithm~\ref{fig:algorithm} to reduce the number of round-trips.}\label{fig:algorithm2}
\end{algorithm*}

We illustrate the operation of the reconciliation algorithm using the example in Figure~\ref{fig:example-dags}; the requests/responses sent in the course of the execution are shown in Figure~\ref{fig:messages}.
Initially, when a connection is established between two replicas, they send each other their heads (Algorithm~\ref{fig:algorithm}, line~\ref{line:send-heads}).
In the example of Figure~\ref{fig:example-dags}, $p$ sends $\langle\mathsf{heads}: \{H(E),H(M)\}\rangle$ to $q$, while $q$ sends $\langle\mathsf{heads}: \{H(G),H(K)\}\rangle$ to $p$.

Each replica also initializes variables $\mathit{sent}$ and $\mathit{recvd}$ to contain the set of messages sent to/received from the other replica within the scope of this particular connection, $\mathit{missing}$ to contain the set of hashes for which we currently lack a message, and $\mathcal{M}_\mathsf{conn}$ to contain a read-only snapshot of this replica's set of messages $\mathcal{M}$ at the time the connection is established (line~\ref{line:init}).
In practice, this snapshot can be implemented using snapshot isolation~\cite{Berenson:1995}.

A replica may concurrently execute several instances of this algorithm using several connections; each connection then has a separate copy of the variables $\mathit{sent}$, $\mathit{recvd}$, $\mathit{missing}$, and $\mathcal{M}_\mathsf{conn}$, while $\mathcal{M}$ is a global variable that is shared between all connections.
Each connection thread executes independently, except for the blocks marked \emph{atomically}, which are executed only by one thread at a time on a given replica.
$\mathcal{M}$ should be maintained in durable storage, while the other variables may be lost in case of a crash.

On receiving the heads from the other replica (line~\ref{line:recv-heads}), the recipient checks whether the recipient's $\mathcal{M}_\mathsf{conn}$ contains a matching message for each hash.
If any hashes are unknown, it replies with a $\mathsf{needs}$ request for the messages matching those hashes (lines~\ref{line:heads-missing} and \ref{line:send-missing}).
In our running example, $p$ needs $H(G)$, while $q$ needs $H(E)$ and $H(M)$.
A replica responds to such a $\mathsf{needs}$ request by returning all the matching messages in a $\mathsf{msgs}$ response (lines~\ref{line:recv-needs}--\ref{line:end-needs}).

On receiving $\mathsf{msgs}$, we first discard any broadcast messages that are not correctly signed (line~\ref{line:msgs-recvd}): the function $\mathrm{check}(m, s)$ returns $\mathsf{true}$ if $s$ is a valid signature over message $m$ by a legitimate replica in the system, and $\mathsf{false}$ otherwise.
For each correctly signed message we then inspect the hashes.
If any predecessor hashes do not resolve to a known message in $\mathcal{M}_\mathsf{conn}$ or $\mathit{recvd}$, the replica sends another $\mathsf{needs}$ request with those hashes (lines~\ref{line:msgs-missing}--\ref{line:msgs-handle-missing}).
In successive rounds of this protocol, the replicas work their way from the heads along the paths of predecessors, until they reach the common ancestors of both replicas' heads.

Eventually, when there are no unresolved hashes, we update the global set $\mathcal{M}$ to reflect the messages we have delivered, perform a topological sort of the graph of received messages to put them in causal order, deliver them to the application in that order, and conclude the protocol run (lines~\ref{line:missing-empty}--\ref{line:finish}).
Once a replica completes reconciliation (by reaching line~\ref{line:finish}), it can conclude that its current set of delivered messages is a superset of the set of delivered messages on the other replica at the start of reconciliation.

When a message $m$ is broadcast, it is also sent as $\langle\mathsf{msgs}: \{m\}\rangle$ on line~\ref{line:eager-send}, and the recipient treats it the same as $\mathsf{msgs}$ received during reconciliation (lines~\ref{line:recv-msgs}--\ref{line:recv-msgs-end}).
Sending messages in this way is not strictly necessary, as the periodic reconciliations will eventually deliver such messages, but broadcasting them eagerly can reduce latency.
Moreover, when a recipient delivers messages, it may also choose to eagerly relay them to other replicas it is connected to, without waiting for the next reconciliation (line~\ref{line:eager-relay}); this also reduces latency, but may result in a replica redundantly receiving messages that it already has.
The literature on gossip protocols examines in detail the question of when replicas should forward messages they receive~\cite{Leitao:2009fi}, while considering trade-offs of delivery latency and bandwidth use; we leave a detailed discussion out of scope for this paper.

We prove in Appendix~\ref{sec:proof} that this algorithm implements all five properties of causal broadcast.
Even though Byzantine-faulty replicas may send arbitrarily malformed messages, a correct replica will not deliver messages without a complete predecessor graph.
Any messages delivered by one correct replica will eventually reach every other correct replica through reconciliations.
After reconciliation, both replicas have delivered the same set of messages.

\subsection{Reducing the number of round trips}\label{sec:algorithm2}

A downside of Algorithm~\ref{fig:algorithm} is that the number of round trips can be up to the length of the longest path in the predecessor graph, making it slow when performing reconciliation over a high-latency network.
We now show how to reduce the number of round-trips using Bloom filters~\cite{Bloom:1970} and a small amount of additional state.

Note that Algorithm~\ref{fig:algorithm} does not store any information about the outcome of the last reconciliation with a particular replica; if two replicas periodically reconcile their states, they need to discover each other's state from scratch on every protocol run.
As per \S~\ref{sec:system-model} we assume that communication between replicas is authenticated, and thus a replica knows the identity of the other replica it is communicating with.
We can therefore record the outcome of a protocol run with a particular replica, and use that information in the next reconciliation with the same replica.
We do this by adding the following instruction after line~\ref{line:update-m} of Algorithm~\ref{fig:algorithm}, where $q$ is the identity of the current connection's remote replica:
\[ \textsc{StoreHeads}(q, \mathrm{heads}(\mathcal{M}_\mathsf{conn} \cup \mathit{recvd})) \]
which updates a key-value store in durable storage, associating the value $\mathrm{heads}(\mathcal{M}_\mathsf{conn} \cup \mathit{recvd})$ with the key $q$ (overwriting any previous value for that key if appropriate).
We use this information in Algorithm~\ref{fig:algorithm2}, which replaces the ``on connecting'' and ``on receiving heads'' functions of Algorithm~\ref{fig:algorithm}, while leaving the rest of Algorithm~\ref{fig:algorithm} unchanged.

First, when replica $p$ establishes a connection with replica $q$, $p$ calls $\textsc{LoadHeads}(q)$ to load the heads from the previous reconciliation with $q$ from the key-value store (Algorithm~\ref{fig:algorithm2}, line~\ref{line:load-heads}).
This function returns the empty set if this is the first reconciliation with $q$.
In Figure~\ref{fig:example-dags}, the previous reconciliation heads might be $\{H(B)\}$.

In lines~\ref{line:msg-since-begin}--\ref{line:msg-since-end} of Algorithm~\ref{fig:algorithm2} we find all of the delivered messages that were added to $\mathcal{M}$ since this last reconciliation (i.e.\ all messages that are not among the last reconciliation's heads or their predecessors), and on line~\ref{line:make-bloom} we construct a Bloom filter~\cite{Bloom:1970} containing those messages.
A Bloom filter is a space-efficient data structure for testing set membership.
It is an array of $m$ bits that is initially all zero; in order to indicate that a certain element is in the set, we choose $k$ bits to set to 1 based on the hash of the element.
To test whether an element is in the set, we check whether all $k$ bits for the hash of that element are set to 1; if so, we say that the element is in the set.
This procedure may produce false positives because it is possible that all $k$ bits were set to 1 due to different elements, not due to the element being checked.
The false-positive probability is a function of the number of elements in the set, the number of bits $m$, and the number of bits $k$ that we set per element~\cite{Bloom:1970,Bose:2008,Christensen:2010}.

We assume $\textsc{MakeBloomFilter}(S)$ creates a Bloom filter from set $S$ and $\textsc{BloomMember}(F,s)$ tests if the element $s$ is a member of the Bloom filter $F$.
In the example of Figure~\ref{fig:example-dags}, $p$'s Bloom filter would contain $\{C, D, E, J, K, L, M\}$, while $q$'s filter contains $\{F, G, J, K\}$.
We send this Bloom filter to the other replica, along with the heads (Algorithm~\ref{fig:algorithm2}, lines~\ref{line:make-bloom}--\ref{line:a2-send-heads}).

On receiving the heads and Bloom filter, we identify any messages that were added since the last reconciliation that are \emph{not} present in the Bloom filter's membership check (line~\ref{line:bloom-member}).
In the example, $q$ looks up $\{F, G, J, K\}$ in the Bloom filter received from $p$; \textsc{BloomMember} returns true for $J$ and $K$. \textsc{BloomMember} is likely to return false for $F$ and $G$, but may return true due to a false positive.
In this example, we assume that \textsc{BloomMember} returns false for $F$ and true for $G$ (a false positive in the case of $G$).

Any Bloom-negative messages are definitely unknown to the other replica, so we send those in reply.
Moreover, we also send any successors of Bloom-negative messages (line~\ref{line:bloom-succ}): since the set $\mathcal{M}$ for a correct replica cannot contain messages whose predecessors are missing, we know that these messages must also be missing from the other replica.
In the example, $q$ sends $\langle\mathsf{msgs}: \{F, G\}\rangle$ to $p$, because $F$ is Bloom-negative and $G$ is a successor of $F$.

Due to Bloom filter false positives, the set of messages in the reply on line~\ref{line:a2-heads-reply} may be incomplete, but it is likely to contain most of the messages that the other replica is lacking.
To fill in the remaining missing messages we revert back to Algorithm~\ref{fig:algorithm}, and perform round trips of $\mathsf{needs}$ requests and $\mathsf{msgs}$ responses until the received set of messages is complete.

The size of the Bloom filter can be chosen dynamically based on the number of elements it contains.
Note that the Bloom filter reflects only messages that were added since the last reconciliation with $q$, not all messages $\mathcal{M}$.
Thus, if the reconciliations are frequent, they can employ a small Bloom filter size to minimize the cost.

This optimized algorithm also tolerates Byzantine faults.
For example, a faulty replica may send a correct replica an arbitrarily corrupted Bloom filter, but this only changes the set of messages in the reply from the correct replica, and has no effect on $\mathcal{M}$ at the correct replica.
We formally analyze the correctness of this algorithm in Appendix~\ref{sec:proof}.

\begin{table*}
\caption{Determining safety of updates with respect to different types of invariant}\label{tab:safety}
\begin{tabular}{l|l}
\toprule
\textbf{Invariant} & \textbf{Update is unsafe if it\dots} \\
\midrule
Row-level check constraint & Inserts/updates tuple with a value that violates the check \\
Attribute has non-negative value & Subtracts a positive amount from the value of that attribute \\
Foreign key constraint & Deletes a tuple from the constraint's target relation \\
Uniqueness of an attribute & Inserts tuple with user-chosen value of that attribute (may be safe if the value is \\ & determined by the hash of the message containing the update) \\
Value is materialized view of a query & All updates are safe, provided materialized view is updated after applying updates \\
\bottomrule
\end{tabular}
\end{table*}

\subsection{Discussion}\label{sec:algorithm-discussion}

One further optimization could be added to our algorithms: on receiving the other replica's heads, a replica can check whether it has any successors of those heads.
If so, those successors can immediately be sent to the other replica (Git calls this a ``fast-forward'').
If neither replica's heads are known to the other (i.e.\ their histories have diverged), they fall back to the aforementioned algorithm.
We have omitted this optimization from our algorithms because we found that it did not noticeably improve the performance of Algorithm~\ref{fig:algorithm2}, and so it was not worth the additional complexity.

A potential issue with Algorithms~\ref{fig:algorithm} and~\ref{fig:algorithm2} is the unbounded growth of storage requirements, since the set $\mathcal{M}$ grows monotonically (much like most algorithms for Byzantine agreement, which produce an append-only log without considering how that log might be truncated).
If the set of replicas in the system is known, we can truncate history as follows: once every replica has delivered a message $m$ (i.e.\ $m$ is \emph{stable}~\cite{Birman:1991el}), the algorithm no longer needs to refer to any of the predecessors of $m$, and so all of those predecessors can be safely removed from $\mathcal{M}$ without affecting the algorithm.
Stability can be determined by keeping track of the latest heads for each replica, and propagating this information between replicas.

When one of the communicating replicas is Byzantine-faulty, the reconciliation algorithm may never terminate, e.g.\ because the faulty replica may send hashes that do not resolve to any message, and so the state $\mathit{missing} = \{\}$ is never reached.
However, in a non-terminating protocol run no messages are delivered and $\mathcal{M}$ is never updated, and so the actions of the faulty replica have no effect on the state of the correct replica.
Reconciliations with other replicas are unaffected, since replicas may perform multiple reconciliations concurrently.

In a protocol run that terminates, the only possible protocol violations from a Byzantine-faulty replica are to omit heads, or to extend the set $\mathcal{M}$ with well-formed messages (i.e.\ messages containing only hashes that resolve to other messages, signed with the private key of one of the replicas in the system).
Any omitted messages will eventually be received through reconciliations with other correct replicas, and any added messages will be forwarded to other replicas; either way, the eventual delivery property of causal broadcast is preserved.

Arbitrary $\mathsf{needs}$ requests sent by a faulty replica do not affect the state of the recipient.
Thus, a faulty replica cannot corrupt the state of a correct replica in a way that would prevent it from later reconciling with another correct replica.

If one of the replicas crashes, both replicas abort the reconciliation and no messages are delivered.
The next reconciliation attempt then starts afresh.
(If desired, it would not be difficult to modify the algorithm so that an in-progress reconciliation can be restarted.)
Note that it is possible for one replica to complete reconciliation and to deliver its new messages while the other replica crashes just before reaching this point.
Thus, when we load the heads from the previous reconciliation on line~\ref{line:load-heads} of Algorithm~\ref{fig:algorithm2}, the local and the remote replica's $\mathit{oldHeads}$ may differ.
This does not affect the correctness of the algorithm.

\begin{algorithm*}
    \begin{algorithmic}[1]
    \On{commit of local transaction $T$}
        \State \textbf{let} $(\mathit{ins}, \mathit{del}) := \Call{GeneratedUpdates}{T}$
        \State \textbf{broadcast} $(\mathit{ins}, \mathit{del})$ by causal broadcast
    \EndOn
    \State
    \On{delivering $m$ by causal broadcast}
        \State \textbf{let} $((\mathit{ins}, \mathit{del}), \mathit{hs}, \mathit{sig}) := m$
        \If{updates $(\mathit{ins}, \mathit{del})$ are safe w.r.t. all invariants \textbf{and}\label{line:safety-check}\\\hspace{3em}$\forall (h,r,t) \in \mathit{del}.\; h \in \{H(m') \mid m' \in \mathrm{pred}^*(\mathcal{M}, m)\}$}\label{line:pred-check}
            \State $S := S \setminus \mathit{del} \,\cup\, \{(H(m), \mathit{rel}, \mathit{tuple}) \mid (\mathit{rel}, \mathit{tuple}) \in \mathit{ins}\}$ \label{line:state-update}
        \EndIf
    \EndOn
    \end{algorithmic}
    \caption{BEC database replication using causal broadcast.}\label{fig:algorithm3}
\end{algorithm*}

\subsection{BEC replication using causal broadcast}\label{sec:implementing-bec}

Given the Byzantine causal broadcast protocol we have defined, we now introduce a replication algorithm that ensures Byzantine Eventual Consistency, assuming \I-confluence.
The details depend on the data model of the database being replicated, and the types of updates allowed.
Algorithm~\ref{fig:algorithm3} shows an approach for a relational database that supports insertion and deletion of tuples in unordered relations (updates are performed by deletion and re-insertion).

Let each replica have state $S$, which is stored durably, and which is initially the empty set.
$S$ is a set of triples: $(h, \mathit{rel}, \mathit{tuple}) \in S$ means the relation named $\mathit{rel}$ contains the tuple $\mathit{tuple}$, and that tuple was inserted by a message whose hash is $h$.
We assume the schema of $\mathit{tuple}$ is known to the application, and we ignore DDL in this example.
Our algorithm gives each replica a full copy of the database $S$; sharding/partitioning could be added if required.

When a transaction $T$ executes, we allow it to read the current state of $S$ at the local replica.
When $T$ commits, we assume the function $\textsc{GeneratedUpdates}(T)$ returns the inserts and deletions performed in $T$.
Insertions are represented as pairs of $(\mathit{rel}, \mathit{tuple})$ indicating the insertion of $\mathit{tuple}$ into the relation named $\mathit{rel}$.
Deletions are represented by the $(h, \mathit{rel}, \mathit{tuple})$ triple to be deleted.
We encode these updates as a message and disseminate it to the other replicas via Byzantine causal broadcast.

When a message is delivered by causal broadcast (including self-delivery to the replica that sent the message), we first check on line~\ref{line:safety-check} of Algorithm~\ref{fig:algorithm3} whether the updates are \emph{safe} with regard to all of the invariants in the application.
An update is unsafe if applying that update could cause the set of transactions to no longer be \I-confluent with regard to a particular invariant.
For example, a deletion of a tuple in a particular relation is unsafe if there is a referential integrity invariant enforcing a foreign key constraint whose target is that relation, because a different transaction could concurrently insert a tuple that has a foreign key reference to the deleted tuple, leading to a constraint violation.
Rules for determining safety for various common types of invariant are described in Table~\ref{tab:safety}; for a deeper analysis of the \I-confluence of different types of constraint, see the extended version of Bailis et al.'s paper~\cite{Bailis:2014ext}.
Note that safety can be determined as a function of only the updates and the invariants, without depending on the current replica state.
The check for safety at the time of delivering a message is necessary because even if we assume that correct replicas only generate safe updates, faulty replicas may generate unsafe updates, which must be ignored by correct replicas.

In addition to checking safety, we check on line~\ref{line:pred-check} that any deletions in the message $m$ are for tuples that were inserted by a message that causally precedes $m$.
This ensures that the insertion is applied before the deletion on all replicas, which is necessary to ensure convergence.
If these conditions are met, we apply the updates to the replica state $S$ on line~\ref{line:state-update}.
We remove any deleted tuples, and we augment any insertions with the hash of the message.
This ensures that subsequent deletions can unambiguously reference the element of $S$ to be deleted, even if multiple replicas concurrently insert the same tuple into the same relation.

We prove in Appendix~\ref{sec:bec-proof} that this algorithm ensures Byzantine Eventual Consistency.
This algorithm could be extended with other operations besides insertion and deletion: for example, it might be useful to support an operation that adds a (possibly negative) value to a numeric attribute; such an operation would commute trivially with other operations of the same type, allowing several concurrent updates to a numeric value (e.g.\ an account balance) to be merged.
Further data models and operations can be implemented using CRDT techniques (\S~\ref{sec:crdts}).

\begin{figure*}
  \includegraphics[width=\textwidth,keepaspectratio=true]{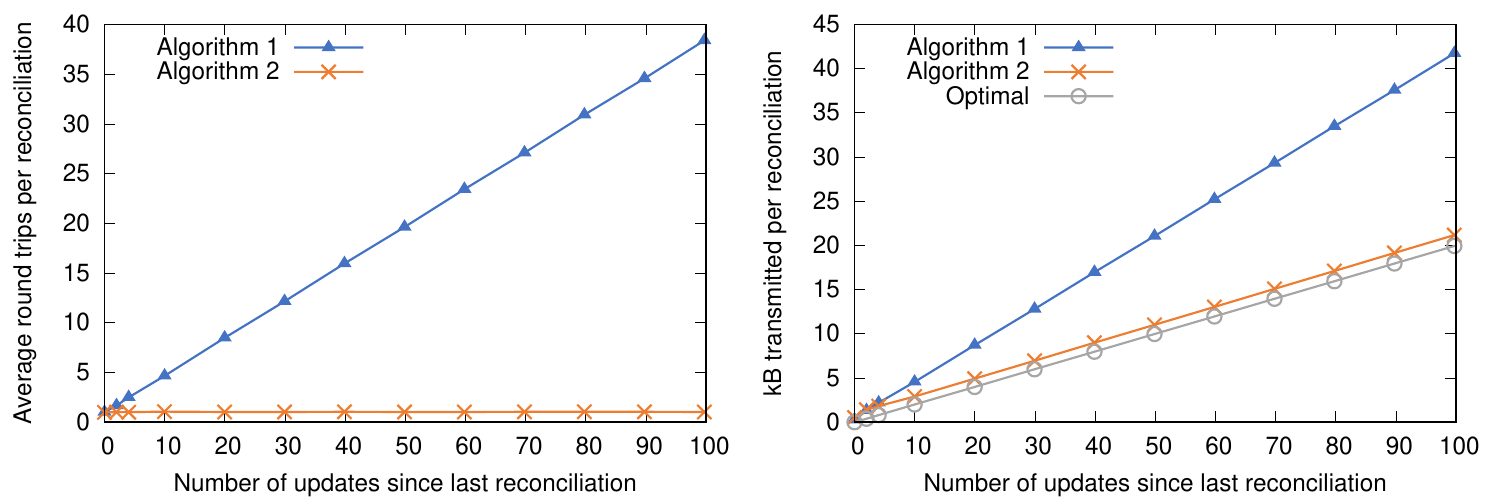}
  \caption{Results from the evaluation of our prototype. Left: number of round trips to complete a reconciliation; right: network bandwidth used per reconciliation (lower is better).}
  \label{fig:evaluation}
\end{figure*}

\section{Evaluation}\label{sec:evaluation}

To evaluate the algorithms introduced in \S~\ref{sec:algorithm} we implemented a prototype and measured its behavior.\footnote{Source code available at \url{https://github.com/ept/byzantine-eventual}}
Our prototype runs all replicas in-memory in a single process, simulating a message-passing network and recording statistics such as the number of messages, hashes, and Bloom filter bits transmitted, in order to measure the algorithm's network communication costs.
In our experiments we use four replicas, each replica generating updates at a constant rate, and every pair of replicas periodically reconciles their states.
We then vary the intervals at which reconciliations are performed (with longer intervals, more updates accumulate between reconciliations) and measure the network communication cost of performing the reconciliation.
To ensure we exercise the reconciliation algorithm, replicas do not eagerly send messages (lines~\ref{line:eager-send} and~\ref{line:eager-relay} of Algorithm~\ref{fig:algorithm} are omitted), and we rely only on periodic reconciliation to exchange messages.
For each of the following data points we compute the average over 600 reconciliations (100 reconciliations between each distinct pair of replicas) performed at regular intervals.

First, we measure the average number of round-trips required to complete one reconciliation (Figure~\ref{fig:evaluation} left).
The greater the number of updates generated between reconciliations, the longer the paths in the predecessor graph.
Therefore, when Algorithm~\ref{fig:algorithm} is used, the number of round trips increases linearly with the number of updates added.
However, Algorithm~\ref{fig:algorithm2} reduces each reconciliation to 1.03 round trips on average, and this number remains constant as the number of updates grows.
96.7\% of reconciliations with Algorithm~\ref{fig:algorithm2} complete in one round trip, 3.2\% require two round trips, and 0.04\% require three or more round trips.
These figures are based on using Bloom filters with 10 bits per entry and 7 hash functions.

Next, we estimate the network traffic resulting from the use of our algorithms.
For this, we assume that each update is 200 bytes in size (not counting its predecessor hashes), hashes are 32 bytes in size (SHA-256), and Bloom filters use 10 bits per element.
Moreover, we assume that each request or response message incurs an additional constant overhead of 100 bytes (e.g.\ for TCP/IP packet headers and signatures).
We compute the number of kilobytes sent per reconciliation (in both directions) using each algorithm.

Figure~\ref{fig:evaluation} (right) shows the results from this experiment.
The gray line represents a hypothetical optimal algorithm that transmits only new messages, but no additional metadata such as hashes or Bloom filters.
Compared to this optimum, Algorithm~\ref{fig:algorithm2} incurs a near-constant overhead of approximately 1~kB per reconciliation for the heads and predecessor hashes, Bloom filter, and occasional additional round trips.
In contrast, the cost of Algorithm~\ref{fig:algorithm} is more than double the optimal, primarily because it sends many $\mathsf{needs}$ messages containing hashes, and it sends messages in many small responses rather than batched into one response.
Thus, we can see that in terms of network performance, Algorithm~\ref{fig:algorithm2} is close to the optimum in terms of both round trips and bytes transmitted, making it viable for use in practice.

In our prototype, all replicas correctly follow the protocol.
Adding faulty replicas may alter the shape of the predecessor graph (e.g.\ resulting in more concurrent updates than there are replicas), but we believe that this would not fundamentally alter our results.
We leave an evaluation of other metrics (e.g.\ CPU or memory use) for future work.

\section{Related Work}\label{sec:relwork}

Hash chaining is widely used: in Git repositories~\cite{GitHTTP}, Merkle trees~\cite{Merkle:1987}, blockchains~\cite{Bano:2019}, and peer-to-peer storage systems such as IPLD~\cite{IPLD}.
Our Algorithm~\ref{fig:algorithm} has similarities to the protocol used by \texttt{git fetch} and \texttt{git push}; to reduce the number of round trips, Git's ``smart'' transfer protocol sends the most recent 32 hashes rather than just the heads~\cite{GitHTTP}.
Git also supports an experimental ``skipping'' reconciliation algorithm~\cite{GitSkipping} in which the search for common ancestors skips some vertices in the predecessor graph, with exponentially growing skip sizes; this algorithm ensures a logarithmic number of round-trips, but may end up unnecessarily transmitting commits that the recipient already has.
Other authors~\cite{Baird:2016tq,Kang:2003} also discuss replicated hash graphs but do not present efficient reconciliation algorithms for bringing replicas up-to-date.

Byzantine agreement has been the subject of extensive research and has seen a recent renewal of interest due to its application in blockchains~\cite{Bano:2019}.
To tolerate $f$ faults, Byzantine agreement algorithms typically require $3f+1$ replicas~\cite{Castro:1999,Kotla:2007,Bessani:2014}, and some even require $5f+1$ replicas~\cite{Abd:2005,Martin:2006}.
This bound can be lowered, for example, to $2f+1$ if synchrony and digital signatures are assumed~\cite{Abraham:2017}. 
Most algorithms also require at least one round of communication with at least $2f+1$ replicas, incurring both significant latency and limiting availability.
Some algorithms instead take a different approach to bounding the number of failures: for example, Upright~\cite{Clement:2009} separates the number of crash failures ($u$) and Byzantine failures ($r$) and uses $2u+r+1$ replicas.
Byzantine quorum systems~\cite{Malkhi:1998} generalize from a threshold $f$ of failures to a set of possible failures.
Zeno~\cite{Singh:2009} makes progress with just $f+1$ replicas, but safety depends on less than $\frac{1}{3}$ of replicas being Byzantine-faulty.
Previous work on Byzantine fault tolerant CRDTs~\cite{Chai:2014,Shoker:2017,Zhao:2016}, Secure Reliable Multicast~\cite{Malki:1996,Malkhi:2000}, Secure Causal Atomic Broadcast~\cite{Cachin:2001cj,Duan:2017} and Byzantine Lattice Agreement~\cite{DiLuna:2020} also assumes $3f+1$ replicas.
All of these algorithms require Sybil countermeasures, such as central control over the participating replicas' identities; moreover, many algorithms ignore the problem of reconfiguring the system to change the set of replicas.

Little prior work tolerates arbitrary numbers of Byzantine-faulty replicas. Depot~\cite{Mahajan:2011} and OldBlue~\cite{VanGundy:2012} provide causal broadcast in this model: OldBlue's algorithm is similar to our Algorithm~\ref{fig:algorithm}, while Depot uses a more complex replication algorithm involving a combination of logical clocks and hash chains to detect and recover from inconsistencies.
We were not able to compare our algorithms to Depot because the available publications~\cite{Mahajan:2010,Mahajan:2011,Mahajan:2012} do not describe Depot's algorithm in sufficient detail to reproduce it.
Depot's consistency model (fork-join-causal) is specific to a key-value data model, and unlike BEC it does not consider the problem of maintaining invariants.
Mahajan et al.\ have also shown that no system that tolerates Byzantine failures can enforce fork causal~\cite{Mahajan:2011} or stronger consistency in an always available, one-way convergent system~\cite{Mahajan:2011cac}. 
BEC provides a weaker two-way convergence property, which requires that eventually a correct replica's updates are reflected on another correct replica only if they can bidirectionally exchange messages for a sufficient period.

Recent work by van der Linde et al.~\cite{vanderLinde:2020} also considers causally consistent replication in the face of Byzantine faults, taking a very different approach to ours: detecting cryptographic proof of faulty behavior, and banning replicas found to be misbehaving.
This approach relies on a trusted central server and trusted hardware such as SGX, whereas we do not assume any trusted components.

In SPORC~\cite{Feldman:2010wl}, BFT2F~\cite{Li:2007} and SUNDR~\cite{Mazieres:2002}, a faulty replica can partition the system, preventing some replicas from ever synchronizing again, so these systems do not satisfy the \emph{eventual update} property of BEC.
Drabkin et al.~\cite{Drabkin:2005} present an algorithm for Byzantine reliable broadcast, but it does not provide causal ordering.

Our reconciliation algorithm is related to the problem of computing the difference, union, or intersection between sets on remote replicas.
This problem has been studied in various domains, including peer-to-peer systems, deduplication of backups, and error-correction.
Approaches include using Bloom filters~\cite{Skjegstad:2011}, invertible Bloom filters~\cite{Goodrich:2011,Eppstein:2011} and polynomial encoding~\cite{Minsky:2003}.
However, these approaches are not designed to tolerate Byzantine faults.

\I-confluence was introduced by Bailis et al.~\cite{Bailis:2014} in the context of non-Byzantine systems.
It is closely related to the concept of logical monotonicity~\cite{Conway:2012} and the CALM theorem~\cite{Ameloot:2013,Hellerstein:2010}, which states that coordination can be avoided for programs that are monotonic.
COPS~\cite{Lloyd:2011} is an example of a non-Byzantine system that achieves causal consistency while avoiding coordination, and BEC is a Byzantine variant of COPS's Causal+ consistency model.
Non-Byzantine causal broadcast was introduced by the ISIS system~\cite{Birman:1991el}.

% Snapdoc~\cite{Kollmann:2019hf} provides cryptographic integrity checks for CRDTs, but its approach (using RSA accumulators) incurs large overheads.
% Truong et al.~\cite{Truong:2012et} present another scheme for authenticating CRDT history, but do not include a reconciliation protocol.

% BAR Gossip uses a model of "rational" replicas that are assumed to maximise a given loss function
% Harry C. Li, Allen Clement, et al. BAR Gossip. OSDI 2006 https://static.usenix.org/event/osdi06/tech/full_papers/li/li.pdf

% Tseng et al.~\cite{Tseng:2019jb} prove that Byzantine causal memory can only be done with $3f+1$ processes?!

% https://github.com/sipa/minisketch

% Comparison to Julien Quintard work on byzantine file systems https://www.repository.cam.ac.uk/bitstream/handle/1810/243442/thesis.pdf?sequence=1&isAllowed=y
% https://infinit.sh

% Comparison to irmin
% https://mirage.github.io/irmin/irmin/Irmin/index.html#syncing-with-a-remote
% https://github.com/mirage/irmin/blob/master/src/irmin/sync_ext.ml#L86-L123
% Send paper to Irmin authors?

\section{Conclusions}\label{sec:conc}

Many peer-to-peer systems tolerate only a bounded number of Byzantine-faulty nodes, and therefore need to employ expensive countermeasures against Sybil attacks, such as proof-of-work, or centrally controlled permissions for joining the system.
In this work we asked the question: what are the limits of what we can achieve without introducing Sybil countermeasures?
In other words, which applications can tolerate arbitrary numbers of Byzantine faults?

We have answered this question with both a positive and a negative result.
Our positive result is an algorithm that achieves Byzantine Eventual Consistency in such a system, provided that the application's transactions are \I-confluent with regard to its invariants.
Our negative result is an impossibility proof showing that such an algorithm does not exist if the application is not \I-confluent.
We proved our algorithms correct, and demonstrated that our optimized algorithm incurs only a small network communication overhead compared to the theoretical optimum, making it immediately applicable in practice.

As shown in \S~\ref{sec:crdts}, many existing systems and applications use CRDTs to achieve strong eventual consistency in a non-Byzantine model.
These applications are already \I-confluent, and adopting our approach will allow those systems to gain robustness against Byzantine faults.
For systems that currently require all nodes to be trusted, and hence can only be deployed in trusted datacenter networks, adding Byzantine fault tolerance opens up new opportunities for deployment in untrusted peer-to-peer settings.

We hope that BEC will inspire further research to ensure the correctness of data systems in the presence of arbitrary numbers of Byzantine faults.
Some open questions include:
\begin{itemize}
    \item How can we best ensure that correct replicas form a connected component, as assumed in \S~\ref{sec:system-model}?
    Connecting each replica to every other is the simplest solution, but it can be expensive if the number of replicas is large.
    \item How can we formalize Table~\ref{tab:safety}, i.e.\ the process of checking whether an update is safe with regard to an invariant?
    \item Is it generally true that a problem can be solved without coordination in a non-Byzantine system if and only if it is immune to Sybil attacks in a Byzantine context?
\end{itemize}

\begin{acks}
Thank you to Alastair Beresford, Jon Crowcroft, Srinivasan Keshav, Smita Vijaya Kumar and Gavin Stark for feedback on a draft of this paper.
Martin Kleppmann is supported by a Leverhulme Trust Early Career Fellowship, the Isaac Newton Trust, and Nokia Bell Labs.
This work was funded in part by EP/T022493/1.
\end{acks}

\bibliographystyle{ACM-Reference-Format}

\appendix

\section{Proofs for causal broadcast}\label{sec:proof}

In this appendix we show that Algorithms~\ref{fig:algorithm} and \ref{fig:algorithm2} implement causal broadcast, as defined in \S~\ref{sec:broadcast}, in the Byzantine system model of \S~\ref{sec:system-model}.
Where a lemma does not specify which of the two algorithms it applies to, it holds for both.

% self-delivery, eventual delivery, authenticity, non-duplication, causal order
\begin{lemma}\label{lemma:easy-properties}
Algorithms~\ref{fig:algorithm} and \ref{fig:algorithm2} satisfy the \emph{self-delivery}, \emph{authenticity}, \emph{non-duplication}, and \emph{causal order} properties of causal broadcast, as defined in \S~\ref{sec:broadcast-properties}.
\end{lemma}
\begin{proof}
The \emph{self-delivery} property holds trivially, because each time a correct replica broadcasts a message, it immediately delivers that message to itself (Algorithm~\ref{fig:algorithm}, line~\ref{line:deliver-local}).

The \emph{authenticity} property holds because when a broadcast message is delivered, it was either sent by the local replica (Algorithm~\ref{fig:algorithm}, line~\ref{line:deliver-local}), in which case it is trivially authentic, or it was received from another replica (Algorithm~\ref{fig:algorithm}, line~\ref{line:deliver}).
In the latter case, messages are in the set $\mathcal{M}$ only if they were broadcast, and we discard any messages that do not have a valid signature from its sender (Algorithm~\ref{fig:algorithm}, line~\ref{line:msgs-recvd}).
Our system model assumes that signatures are unforgeable, so a correct replica delivers a message only if it was broadcast by the replica that signed it.

The \emph{non-duplication} property holds because every message is unique (due to the inclusion of the hashes of its predecessors), and messages delivered during reconciliation are limited to those not already in $\mathcal{M}$ (Algorithm~\ref{fig:algorithm}, line~\ref{line:deliver}).
Since $\mathcal{M}$ is immediately updated to include all delivered messages, and this takes place in an \emph{atomic} block (preventing two threads from concurrently trying to deliver the same message), this algorithm ensures that no correct replica delivers the same message more than once.

The \emph{causal order} property holds because when a correct replica broadcasts a message, the predecessor hashes are computed such that every message previously broadcast or delivered by this replica becomes a (direct or indirect) predecessor of the new message (Algorithm~\ref{fig:algorithm}, line \ref{line:broadcast-heads}).
Any correct replica delivers messages in topologically sorted order, i.e.\ any predecessors of $m$ are delivered before $m$ (Algorithm~\ref{fig:algorithm}, line \ref{line:deliver}).
The reconciliation algorithm delivers messages only once all hashes have been resolved (once all direct and indirect predecessor messages have been received), so we know that there are no missing predecessors.
Thus, whenever a correct replica broadcasts or delivers $m_1$ before broadcasting $m_2$, all correct replicas deliver $m_1$ before delivering $m_2$.
\end{proof}

This leaves the \emph{eventual delivery} property, which is the focus of the remainder of this appendix.
We consider two correct replicas $p$ and $q$, with initial sets of messages $\mathcal{M}_p$ and $\mathcal{M}_q$ respectively at the start of the execution.
Assume that in this run of the algorithm, $p$ and $q$ both complete the reconciliation by reaching line~\ref{line:finish} of Algorithm~\ref{fig:algorithm}.
Let $\mathit{recvd}_p$ be the contents of the variable $\mathit{recvd}$ at replica $p$ when the reconciliation is complete, and likewise $\mathit{recvd}_q$ at replica $q$.
Further, let $\mathcal{M}'_p = \mathcal{M}_p \cup \mathit{recvd}_p$ and $\mathcal{M}'_q = \mathcal{M}_q \cup \mathit{recvd}_q$ be the final set of messages at both replicas.

\begin{lemma}\label{lemma:no-p-missing}
The set of messages $\mathcal{M}$ of a correct replica $p$ grows monotonically.
\end{lemma}
\begin{proof}
The replica $p$ only modifies $\mathcal{M}$ by generating new operations, which are added to $\mathcal{M}$ (Algorithm~\ref{fig:algorithm}, line~\ref{line:update-m-local}), or by unioning it with the set $\mathit{recvd}$ (Algorithm~\ref{fig:algorithm}, line~\ref{line:update-m}).
Thus, elements are only added to the set $\mathcal{M}$, and therefore $\mathcal{M}$ grows monotonically.
\end{proof}

\begin{lemma}\label{lemma:no-dangling}
Let $m = (v, \mathit{hs}, \mathit{sig})$ such that $m \in \mathcal{M}_p$.
Then its predecessors are also in $\mathcal{M}_p$, i.e.\ $\forall h \in \mathit{hs}.\; \exists m' \in \mathcal{M}_p.\; H(m') = h$.
\end{lemma}
\begin{proof}
There are two ways $m$ can become a member of $\mathcal{M}_p$ for a correct replica $p$:
\begin{description}
    \item[Case] $m$ is broadcast by replica $p$:\\
    In this case, since $p$ is assumed to be correct, the hashes $\mathit{hs}$ are computed as $\mathit{hs} = \{H(m') \mid m' \in \mathcal{M} \wedge \mathrm{succ}^1(\mathcal{M}, m') = \{\}\,\}$ for some earlier state $\mathcal{M}$ (Algorithm~\ref{fig:algorithm}, line~\ref{line:broadcast-heads}).
    As $\mathcal{M}$ grows monotonically (Lemma~\ref{lemma:no-p-missing}), $\mathcal{M} \subseteq \mathcal{M}_p$, and thus we can deduce that $\forall h \in \mathit{hs}.\; \exists m' \in \mathcal{M}_p.\; H(m') = h$.
    \item[Case] $m$ is received from another replica (which may be faulty):\\
    During the run of the protocol at which $p$ received $m$, we have $m \in \mathit{recvd}$ and $\mathit{missing} = \{\}$ at line~\ref{line:update-m} of Algorithm~\ref{fig:algorithm}.
    Let $\mathcal{M}$ be the set of messages at $p$ immediately before that execution of line~\ref{line:update-m}.
    From $\mathit{missing} = \{\}$ and line~\ref{line:msgs-missing} of Algorithm~\ref{fig:algorithm} we have $\forall h \in \mathit{hs}.\; \exists m' \in (\mathcal{M} \cup \mathit{recvd}).\; H(m') = h$.
    Since $\mathcal{M}$ grows monotonically (Lemma~\ref{lemma:no-p-missing}) and $\mathit{recvd} \subseteq \mathcal{M}_p$ (Algorithm~\ref{fig:algorithm}, line~\ref{line:update-m}), $\forall h \in \mathit{hs}.\; \exists m' \in \mathcal{M}_p.\; H(m') = h$.
\end{description}
\end{proof}

\begin{lemma}\label{lemma:no-collision}
Let $m = (v, \mathit{hs}, \mathit{sig})$ such that $m \in \mathcal{M}_p$ and $m \in \mathcal{M}_q$.
Then the hashes $\mathit{hs}$ resolve to the same messages at $p$ and $q$, that is, $\{m' \in \mathcal{M}_p \mid H(m') \in \mathit{hs}\} = \{m' \in \mathcal{M}_q \mid H(m') \in \mathit{hs}\}$.
\end{lemma}
\begin{proof}
We use proof by contradiction.
Assume there exists $h \in \mathit{hs}$ such that $\{m' \in \mathcal{M}_p \mid H(m') = h\} \neq \{m' \in \mathcal{M}_q \mid H(m') = h\}$.
By Lemma~\ref{lemma:no-dangling} we have $\{m' \in \mathcal{M}_p \mid H(m') = h\} \neq \{\}$ and similarly, $\{m' \in \mathcal{M}_q \mid H(m') = h\} \neq \{\}$.
Hence, there exist $m' \in \mathcal{M}_p$ and $m'' \in \mathcal{M}_q$ such that $m' \neq m''$ and $H(m') = H(m'') = h$.
However, this contradicts our assumption in \S~\ref{sec:algorithm} that the hash function $H(\cdot)$ is collision-resistant.
\end{proof}

\begin{lemma}\label{lemma:no-q-missing}
$\mathcal{M}_q \subseteq \mathcal{M}'_p$ when executing Algorithm~\ref{fig:algorithm}.
\end{lemma}
\begin{proof}
We use proof by contradiction.
Assume to the contrary that $\exists m \in \mathcal{M}_q.\; m \notin  \mathcal{M}'_p$.
Since $\mathit{recvd} \subseteq \mathcal{M}'_p$ and elements are only added to $\mathit{recvd}$ (Algorithm~\ref{fig:algorithm}, line~\ref{line:msgs-recvd}) then $m \notin  \mathcal{M}'_p$ implies that $m \notin \mathit{recvd}$ on replica $p$.
We now consider two cases depending on the value returned by $\mathrm{succ}^1(\mathcal{M}_q, m)$:
\begin{description}
    \item[Case] $\mathrm{succ}^1(\mathcal{M}_q, m) = \{\}$:\\
    In this case, $H(m) \in \mathrm{heads}(\mathcal{M}_q)$, and so the first $\mathsf{heads}$ request from $q$ to $p$ will contain $H(m)$ (Algorithm~\ref{fig:algorithm}, line~\ref{line:send-heads}).
    Since $m \notin \mathcal{M}_p$, replica $p$ will send a $\mathsf{needs}$ request to $q$ containing $H(m)$ (Algorithm~\ref{fig:algorithm}, line~\ref{line:heads-missing}).
    Upon receiving the $\mathsf{needs}$ message containing $H(m)$, replica $q$ will reply with an $\mathsf{msgs}$ response containing $m$ (Algorithm~\ref{fig:algorithm}, line~\ref{line:send-msgs}).
    Replica $p$ will receive the $\mathsf{msgs}$ response with $m$ from replica $q$ and will add $m$ to $\mathit{recvd}$ (Algorithm~\ref{fig:algorithm}, line~\ref{line:msgs-recvd}).
    This contradicts our previous finding that $m \notin \mathit{recvd}$.
    
    \item[Case] $\mathrm{succ}^1(\mathcal{M}_q, m) \ne \{\}$:\\
    In this case, $H(m) \notin \mathrm{heads}(\mathcal{M}_q)$.
    Since $\mathcal{M}_q$ is a DAG, there must exist a message $m'$ such that $H(m') \in \mathrm{heads}(\mathcal{M}_q)$ and $m' \in \mathrm{succ}^*(\mathcal{M}_q, m)$.
    As in the previous case, $H(m') \in \mathrm{heads}(\mathcal{M}_q)$ implies that $m' \in \mathit{recvd}$.
    Note that none of the messages in $\mathrm{succ}^*(\mathcal{M}_q, m)$ are in $\mathcal{M}_p$ as $m \notin \mathcal{M}'_p$ implies that  $m \notin \mathcal{M}_p$ (Lemma~\ref{lemma:no-p-missing}).
    If $m' \in \mathrm{succ}^1(\mathcal{M}_q, m)$ then it must the case that $m \in \mathit{recvd}$ by the time that $\mathit{missing} = \emptyset$, otherwise $m \in \mathit{missing}$ (Algorithm~\ref{fig:algorithm}, line~\ref{line:msgs-missing}).
    By induction over the path of successors from $m'$ to $m$, we observe that $m \in \mathit{recvd}$.
    At each step of the induction, the replicas move to the predecessors of the previous step; due to Lemma~\ref{lemma:no-collision}, $p$ and $q$ agree about the identity of these predecessors.
    This contradicts our previous finding that $m \notin \mathit{recvd}$.
\end{description}
\end{proof}

\begin{lemma}\label{lemma:no-q-missing2}
$\mathcal{M}_q \subseteq \mathcal{M}'_p$ when executing Algorithm~\ref{fig:algorithm2}.
\end{lemma}
\begin{proof}
We use proof by contradiction.
Assume $m \in \mathcal{M}_q$ such that $m \notin \mathcal{M}_p'$ and $m \notin \mathit{recvd}$ like in Lemma~\ref{lemma:no-q-missing}.
Let $\mathit{filter}$ be the Bloom filter in the initial message from $p$ to $q$ in the current protocol run (Algorithm~\ref{fig:algorithm2}, line~\ref{line:make-bloom}).
Even though $m \notin \mathcal{M}_p$ (by Lemma~\ref{lemma:no-p-missing}), $\textsc{BloomMember}(\mathit{filter}, m)$ may return a false positive.
Moreover, if it returns true, $m$ may or may not be a successor of a $\mathit{bloomNegative}$ item as computed in Algorithm~\ref{fig:algorithm2}, lines~\ref{line:bloom-member}--\ref{line:bloom-succ}.
As a result it is possible that either $m \in \mathit{reply}$ or $m \notin \mathit{reply}$ after $q$ has executed line~\ref{line:bloom-succ} of Algorithm~\ref{fig:algorithm2}.

If $m \in \mathit{reply}$ then $p$ will receive an $\mathsf{msgs}$ response containing $m$ from $q$, which will be added to $\mathit{recvd}$, contradicting our assumption that $m \notin \mathit{recvd}$.
If $m \notin \mathit{reply}$ we continue to line~\ref{line:a2-heads-missing} of Algorithm~\ref{fig:algorithm2}, from which point onward the algorithm is the same as Algorithm~\ref{fig:algorithm}.
Thus, we have $\mathcal{M}_q \subseteq \mathcal{M}'_p$ by Lemma~\ref{lemma:no-q-missing}.
\end{proof}

\begin{lemma}\label{lemma:no-extras}
$\mathcal{M}'_p \subseteq \mathcal{M}_p \cup \mathcal{M}_q$.
\end{lemma}
\begin{proof}
We use proof by contradiction.
Assume to the contrary that $\exists m \in \mathcal{M}'_p.\; m \notin \mathcal{M}_p  \land  m \notin \mathcal{M}_q$.
Since $\exists m \in \mathcal{M}'_p$, replica $p$ must have received a message containing $m$ from replica $q$ before it completed reconciliation (Algorithm~\ref{fig:algorithm}, lines \ref{line:recv-msgs}--\ref{line:msgs-handle-missing} and \ref{line:update-m}).
Replica $q$ will only send a message containing $m$ if $m \in \mathcal{M}_q$ or $m \in \mathcal{M}'_q$, depending on whether replica $q$ has completed the reconciliation algorithm.
Since $m \notin \mathcal{M}_q$ then replica $q$ must have received a message containing $m$ from replica $p$.
Since $m \notin \mathcal{M}_p$ then replica $p$ will not send this message and therefore $m$ does not exist.
\end{proof}

\begin{lemma}\label{lemma:reconcile-equal}
When two correct replicas $p$ and $q$, with initial sets of messages $\mathcal{M}_p$ and $\mathcal{M}_q$, have completed reconciliation (i.e.\ both have reached line~\ref{line:finish} of Algorithm~\ref{fig:algorithm}), then their final sets of messages $\mathcal{M}'_p$ and $\mathcal{M}'_q$  are both equal to $\mathcal{M}_p \cup \mathcal{M}_q$.
\end{lemma}
\begin{proof}
We have $\mathcal{M}_p \subseteq \mathcal{M}'_p$ by Lemma \ref{lemma:no-p-missing}, $\mathcal{M}_q \subseteq \mathcal{M}'_p$ by Lemmata \ref{lemma:no-q-missing} and \ref{lemma:no-q-missing2}, and $\mathcal{M}'_p \subseteq \mathcal{M}_p \cup \mathcal{M}_q$ by Lemma \ref{lemma:no-extras}.
From these facts we have shown that $\mathcal{M}'_p = \mathcal{M}_q \cup \mathcal{M}_q$.
Similarly, by swapping $p$ and $q$ we can show that $\mathcal{M}'_q = \mathcal{M}_q \cup \mathcal{M}_q$.
\end{proof}

\begin{lemma}\label{lemma:termination}
If two correct replicas attempt reconciliation an infinite number of times, then there is an infinite number of protocol runs in which the algorithm terminates (i.e.\ both replicas reach line~\ref{line:finish} of Algorithm~\ref{fig:algorithm}), assuming the system model of \S~\ref{sec:system-model}.
\end{lemma}
\begin{proof}
Our system model assumes network communication is reliable, but it allows correct replicas to crash and recover.
If a crash occurs during a reconciliation, that reconciliation may be aborted.
However, if we perform reconciliation an infinite number of times, then only a subset of reconciliations will be affected by crashes, and an infinite number of reconciliations will be free from crashes.

The graph of messages $\mathcal{M}_p$ at any correct replica $p$ is finite and contains no cycles.
Therefore, every vertex $m' \in \mathcal{M}_p$ can be reached in a finite number of steps by starting a graph traversal at $\mathrm{heads}(\mathcal{M}_p)$ and, in each step, moving from each vertex to its predecessors.
Moreover, by Lemma~\ref{lemma:no-dangling}, $\mathcal{M}_p$ at any correct replica $p$ contains only hashes that are the hash of another message in $\mathcal{M}_p$.
Hence, in a connection in which neither replica crashes, the algorithm will always reach the state $\mathit{missing} = \{\}$ and terminate (i.e.\ reach line~\ref{line:finish} of Algorithm~\ref{fig:algorithm}) in a finite number of round-trips of $\mathsf{needs}$ requests and $\mathsf{msgs}$ responses.

Since there are an infinite number of connection attempts that are free from crashes, and the algorithm always terminates for these connections, we can conclude that there are an infinite number of protocol runs in which the algorithm terminates.
\end{proof}

\begin{theorem}
Algorithms~\ref{fig:algorithm} and~\ref{fig:algorithm2} implement causal broadcast, as defined in \S~\ref{sec:broadcast}, in the Byzantine system model of \S~\ref{sec:system-model}.
\end{theorem}
\begin{proof}
Lemma~\ref{lemma:easy-properties} proves the properties apart from \emph{eventual delivery}.
To prove eventual delivery, for any two correct replicas $p$ and $q$, we must show that a message delivered by $p$ will also be delivered by $q$.
We assume in \S~\ref{sec:system-model} that the correct replicas form a connected component in the graph of replicas and network links.
Thus, either there is a direct network link between $p$ and $q$, or there is a path of network links on which all of the intermediate links are also correct.

Assume that any two adjacent replicas on this path periodically attempt a reconciliation without a bound on the number of reconciliations.
Thus, in an execution of infinite duration, there will be an infinite number of reconciliations between any two adjacent replicas.
By Lemma~\ref{lemma:termination}, an infinite number of these reconciliations will complete.
By Lemma~\ref{lemma:reconcile-equal}, at the instant in which one of these reconciliations completes, the set of messages delivered by one replica equals the set of messages delivered by the other replica, with the exception of any messages delivered by concurrent reconciliations.

Any messages delivered by one replica while the reconciliation with the other replica was in progress will be sent in the next reconciliation, which always exists, since we are assuming an infinite number of reconciliations.
After a reconciliation where one replica completes while the other does not (e.g.\ due to a crash just before completing), the sets of messages delivered by the two replicas may be different, but again the missing messages will be sent in the next reconciliation.

Let $m$ be a message that has been delivered by $p$ at some point in time.
We have $m \in \mathcal{M}_p$ from that time onward, since $\mathcal{M}_p$ is exactly the set of delivered messages, and it grows monotonically (Lemma~\ref{lemma:no-p-missing}).
Thus, $m$ will eventually be delivered by any correct replica to which $p$ has a direct network link.
These replicas will eventually relay $m$ to their direct neighbors, and so on, until $m$ is delivered to $q$ through successive reconciliations along the path from $p$ to $q$.
Therefore, $m$ is eventually delivered by $q$.
\end{proof}

\section{Proofs for BEC replication}\label{sec:bec-proof}

In this appendix we prove that Algorithm~\ref{fig:algorithm3} provides Byzantine Eventual Consistency, as defined in \S~\ref{sec:bec-definition}.

\begin{lemma}\label{lemma:safety-check}
For a given message $m$, the conditions on lines~\ref{line:safety-check} and~\ref{line:pred-check} of Algorithm~\ref{fig:algorithm3} evaluate to the same result ($\mathsf{true}$ or $\mathsf{false}$) on all correct replicas.
\end{lemma}

\begin{proof}
The safety check on line~\ref{line:safety-check} depends only on the updates in the message and the invariants, and not on the replica state.
We assume that all correct replicas use the same invariants, so they will agree on whether the updates are safe.

The predecessor check on line~\ref{line:pred-check} depends only on $m$ and the set of its predecessors; this set is the same on all correct replicas, since it is unambiguously determined by $m$'s causal dependency hashes, and a replica ensures that if a message is in $\mathcal{M}$, its predecessors are also in $\mathcal{M}$.
\end{proof}

\begin{lemma}\label{lemma:commutative}
In Algorithm~\ref{fig:algorithm3}, for all replica states $S$, and for all messages $m_1$, $m_2$ such that neither message is a predecessor of the other, delivering $m_1$ first and then $m_2$ has the same effect as delivering $m_2$ first and then $m_1$.
\end{lemma}

\begin{proof}
Consider first the conditions on lines~\ref{line:safety-check} and~\ref{line:pred-check} of Algorithm~\ref{fig:algorithm3}; by Lemma~\ref{lemma:safety-check}, these conditions always evaluate to the same value for a given message, regardless of the replica state and the order in which the messages are delivered.
If they evaluate to $\mathsf{false}$ for both $m_1$ and $m_2$, the state $S$ is not modified, so these updates commute.
If the conditions are $\mathsf{false}$ for one message and $\mathsf{true}$ for the other, only one of the message deliveries updates the replica state, so these updates commute.

Now consider the case when the conditions are $\mathsf{true}$ for both messages.
Let the following:
\begin{align*}
    m_1 &= ((\mathit{ins}_1, \mathit{del}_1), \mathit{hs}_1, \mathit{sig}_1) \\
    m_2 &= ((\mathit{ins}_2, \mathit{del}_2), \mathit{hs}_2, \mathit{sig}_2) \\
    \mathit{ins}_1' &= \{(H(m_1), \mathit{rel}, \mathit{tuple}) \mid (\mathit{rel}, \mathit{tuple}) \in \mathit{ins}_1\} \\
    \mathit{ins}_2' &= \{(H(m_2), \mathit{rel}, \mathit{tuple}) \mid (\mathit{rel}, \mathit{tuple}) \in \mathit{ins}_2\}
\end{align*}
Since the condition on line~\ref{line:pred-check} is true, we have
\[ \forall (h,r,t) \in \mathit{del}_1.\; h \in \{H(m') \mid m' \in \mathrm{pred}^*(\mathcal{M}, m_1)\}. \]
Recall that $m_2 \notin \mathrm{pred}^*(\mathcal{M}, m_1)$ by assumption, from which we can deduce that $\forall (h,r,t) \in \mathit{del}_1.\; h \ne H(m_2).$
On the other hand, every element in $\mathit{ins}_2'$ contains $H(m_2)$, so $\mathit{ins}_2'$ cannot have any elements in common with $\mathit{del}_1$, that is, $\mathit{ins}_2' \cap \mathit{del}_1 = \{\}$.
By similar argument, $\mathit{ins}_1' \cap \mathit{del}_2 = \{\}$.

A replica that first delivers $m_1$ and then $m_2$ is in the final state
\begin{align*}
S' &= (((S \setminus \mathit{del}_1) \cup \mathit{ins}_1') \setminus \mathit{del}_2) \cup \mathit{ins}_2' \\
&= (((S \setminus \mathit{del}_1) \setminus \mathit{del}_2) \cup \mathit{ins}_1') \cup \mathit{ins}_2' \quad\text{since } \mathit{ins}_1' \cap \mathit{del}_2 = \{\} \\
&= (((S \setminus \mathit{del}_2) \setminus \mathit{del}_1) \cup \mathit{ins}_2') \cup \mathit{ins}_1' \quad\text{by commutativity} \\
&= (((S \setminus \mathit{del}_2) \cup \mathit{ins}_2') \setminus \mathit{del}_1) \cup \mathit{ins}_1' \quad\text{since } \mathit{ins}_2' \cap \mathit{del}_1 = \{\}
\end{align*}
which equals the final state of a replica that first delivers $m_2$ and then $m_1$.
Therefore, the updates by $m_1$ and $m_2$ commute.
\end{proof}

Note that these proofs do not make any assumptions about the sender of a message, and therefore they hold for any messages, even those sent by a malicious replica.

\begin{theorem}
Assume that the safety check on line~\ref{line:safety-check} of Algorithm~\ref{fig:algorithm3} only allows updates from transactions that are \I-confluent with regard to all of the application's declared invariants.
Then Algorithm~\ref{fig:algorithm3} ensures Byzantine Eventual Consistency, as defined in \S~\ref{sec:bec-definition}, in the Byzantine system model of \S~\ref{sec:system-model}.
\end{theorem}

\begin{proof}
The \emph{self-update}, \emph{authenticity}, and \emph{causal consistency} properties follow directly from the \emph{self-delivery}, \emph{authenticity}, and \emph{causal order} properties of causal broadcast (\S~\ref{sec:broadcast-properties}), respectively.

The \emph{atomicity} property of BEC follows from the fact that all updates generated by a transaction are encoded into a single message, and the delivery of that message occurs within an \emph{atomic} block on every correct replica (Algorithm~\ref{fig:algorithm}, line~\ref{line:deliver}).

To prove the \emph{eventual update} update property of BEC, we rely upon the \emph{eventual delivery} property of causal broadcast, which ensures that any message delivered by one correct replica will eventually be delivered by all correct replicas.
Moreover, by Lemma~\ref{lemma:safety-check}, one correct replica applies the updates in a message if and only if another correct replica applies it.

To prove the \emph{convergence} property of BEC, assume two correct replicas $p$ and $q$ have delivered the same set of messages.
We must then show that both replicas are in the same state.
Due to the \emph{non-duplication} property of causal broadcast, no message is delivered more than once.
Therefore, the sequence of message deliveries at $p$ is a permutation of the sequence of deliveries at $q$.
Moreover, due to the \emph{causal order} property, if both replicas have delivered $m_1$ and $m_2$, and if $m_1$ is a predecessor of $m_2$, then $m_1$ appears before $m_2$ in both replicas' delivery sequences.

It is therefore possible to permute $p$'s delivery sequence into $q$'s delivery sequence by repeatedly swapping adjacent messages in that sequence, such that neither message is a predecessor of the other~\cite{Gomes:2017gy}.
By Lemma~\ref{lemma:commutative}, swapping two such adjacent message deliveries does not change the replica state after delivering those messages.
Thus, each such swap of adjacent messages leaves the final state of the replica unchanged, and therefore $p$ and $q$ must be in the same state.

Finally, to prove the \emph{invariant preservation} property of BEC we must show that correct replicas are always in a state that satisfies all of the application's declared invariants.
We use proof by induction over the sequence of messages delivered by a correct replica.
The base case is a replica with an empty database, where no messages have been delivered, and where all invariants are satisfied by assumption.
For the inductive step, assume a replica that has delivered messages $m_1, \dots, m_n$ in that order, and that is in a state in which all invariants are satisfied.
We must then show that after delivering one more message, $m_{n+1}$, the invariants continue to be satisfied.

Split the sequence $m_1, \dots, m_n$ into a prefix $m_1, \dots, m_{i-1}$ and a suffix $m_i, \dots, m_n$ such that the suffix is the longest possible suffix in which all messages are concurrent with $m_{n+1}$ (that is, either the prefix is empty, or $m_{i-1}$ is a predecessor of $m_{n+1}$).
By Lemma~\ref{lemma:commutative}, the replica state after delivering $m_1, \dots, m_n, m_{n+1}$ is the same as the state after delivering $m_1, \dots, m_{n+1}, m_n$ in that order.
By repeated pairwise swaps of concurrent messages, that state is the same as after delivering $m_1, \dots, m_{i-1}, m_{n+1}, m_i, \dots, m_n$.
Immediately before $m_i$ is the earliest possible point at which $m_{n+1}$ can be delivered by causal broadcast.

Let $S$ be the state of the replica after delivering $m_1, \dots, m_{i-1}$, and let $I$ be one of the application's declared invariants.
We have $I(S)$ by the inductive hypothesis.
By assumption, the safety check on line~\ref{line:safety-check} of Algorithm~\ref{fig:algorithm3} does not allow updates that would cause the immediate violation of an invariant.
Moreover, by the definition of \I-confluence (\S~\ref{sec:confluence}), if $I(\mathrm{apply}(S, m_i))$ and $I(\mathrm{apply}(S, m_{n+1}))$ for concurrent $m_i$ and $m_{n+1}$, then the invariant is also satisfied after both $m_i$ and $m_{n+1}$ have been delivered: $I(\mathrm{apply}(\mathrm{apply}(S, m_i), m_{n+1}))$.
By repeated swapping of adjacent messages in this order we move $m_{n+1}$ back to the end of the sequence, at each step ensuring that $I$ continues to be satisfied due to \I-confluence.
This shows that $I$ must still be satisfied after delivering $m_1, \dots, m_n, m_{n+1}$, as required:
\[ I(\mathrm{apply}(\mathrm{apply}(\dots\mathrm{apply}(\{\}, m_1), \dots, m_n), m_{n+1})). \]
\end{proof}

\end{document}

%% file: figs/connected.tikz
\begin{tikzpicture}
\tikzstyle{correct}=[draw,circle,inner sep=3pt]
\tikzstyle{faulty}=[draw,circle,inner sep=3pt,fill=red!30]
\node [faulty] (a) at (0.04,0.70) {};
\node [correct] (b) at (0.22,1.80) {};
\node [faulty] (c) at (0.87,0.47) {};
\node [correct] (d) at (1.72,0.86) {};
\node [faulty] (e) at (1.65,0.09) {};
\node [correct] (f) at (2.78,1.62) {};
\node [faulty] (g) at (2.93,0.73) {};
\node [correct] (h) at (3.42,1.64) {};
\node [faulty] (i) at (3.76,0.36) {};
\node [correct] (j) at (3.98,1.00) {};
\node [faulty] (k) at (4.69,1.04) {};
\node [correct] (l) at (5.14,1.63) {};
\node [faulty] (m) at (5.15,0.40) {};
\node [faulty] (n) at (5.41,1.00) {};
\node [correct] (o) at (6.36,0.93) {};
\node [correct] (p) at (6.89,0.33) {};
\node [faulty] (q) at (7.10,1.61) {};
\node [faulty] (r) at (7.62,0.68) {};
\node [faulty] (s) at (7.90,1.37) {};
\draw (a) -- (b);
\draw (a) -- (c);
\draw (a) -- (d);
\draw (b) -- (d);
\draw (b) -- (f);
\draw (c) -- (d);
\draw (c) -- (e);
\draw (d) -- (f);
\draw (d) -- (g);
\draw (e) -- (g);
\draw (f) -- (h);
\draw (g) -- (i);
\draw (g) -- (j);
\draw (h) -- (j);
\draw (h) -- (l);
\draw (i) -- (j);
\draw (i) -- (m);
\draw (j) -- (k);
\draw (k) -- (l);
\draw (k) -- (m);
\draw (k) -- (n);
\draw (l) -- (o);
\draw (l) -- (q);
\draw (m) -- (p);
\draw (n) -- (o);
\draw (o) -- (p);
\draw (o) -- (q);
\draw (p) -- (r);
\draw (q) -- (s);
\draw (q) -- (r);
\draw (r) -- (s);
\end{tikzpicture}

\vspace{0.5cm}
\begin{tikzpicture}
\tikzstyle{correct}=[draw,circle,inner sep=3pt]
\tikzstyle{faulty}=[draw,circle,inner sep=3pt,fill=red!30]
\node [correct] (a) at (0.04,0.70) {};
\node [correct] (b) at (0.22,1.80) {};
\node [correct] (c) at (0.87,0.47) {};
\node [correct] (d) at (1.72,0.86) {};
\node [faulty] (e) at (1.65,0.09) {};
\node [correct] (f) at (2.78,1.62) {};
\node [faulty] (g) at (2.93,0.73) {};
\node [faulty] (h) at (3.42,1.64) {};
\node [faulty] (i) at (3.76,0.36) {};
\node [correct] (j) at (3.98,1.00) {};
\node [faulty] (k) at (4.69,1.04) {};
\node [correct] (l) at (5.14,1.63) {};
\node [faulty] (m) at (5.15,0.40) {};
\node [faulty] (n) at (5.41,1.00) {};
\node [correct] (o) at (6.36,0.93) {};
\node [correct] (p) at (6.89,0.33) {};
\node [faulty] (q) at (7.10,1.61) {};
\node [faulty] (r) at (7.62,0.68) {};
\node [faulty] (s) at (7.90,1.37) {};
\draw (a) -- (b);
\draw (a) -- (c);
\draw (a) -- (d);
\draw (b) -- (d);
\draw (b) -- (f);
\draw (c) -- (d);
\draw (c) -- (e);
\draw (d) -- (f);
\draw (d) -- (g);
\draw (e) -- (g);
\draw (f) -- (h);
\draw (g) -- (i);
\draw (g) -- (j);
\draw (h) -- (j);
\draw (h) -- (l);
\draw (i) -- (j);
\draw (i) -- (m);
\draw (j) -- (k);
\draw (k) -- (l);
\draw (k) -- (m);
\draw (k) -- (n);
\draw (l) -- (o);
\draw (l) -- (q);
\draw (m) -- (p);
\draw (n) -- (o);
\draw (o) -- (p);
\draw (o) -- (q);
\draw (p) -- (r);
\draw (q) -- (s);
\draw (q) -- (r);
\draw (r) -- (s);
\path (f.north) node [above] {$p$};
\path (l.north) node [above] {$q$};
\end{tikzpicture}

%% file: figs/trivial1.tikz
\begin{tikzpicture}
% Space between timelines
\def\width{1}
% Message delay
\def\delay{0.7}

% Timelimes
\node (p-start) at (0, 0.5) {$p$};
\node (p-end)   at (0, -1.3) {};
\node (q-start) at (\width, 0.5) {$q$};
\node (q-end)   at (\width, -1.3) {};
\node (r-start) at (\width*2, 0.5) {$r$};
\node (r-end)   at (\width*2, -1.3) {};
\draw (p-start) -- (p-end);
\draw (q-start) -- (q-end);
\draw (r-start) -- (r-end);

% Messages
\draw[thick,->] (\width, 0) to node [above,pos=0.4,sloped] {$A$} (0, -\delay) node [left] {$\mathcal{M}_p = \{A\}$};

\draw[thick,->] (\width, -0.1) to node [above,pos=0.4,sloped] {$B$} (\width*2, -0.1-\delay) node [right] {$\mathcal{M}_r = \{B\}$};

\end{tikzpicture}

% \begin{tikzpicture}
% % Timelimes
% \node (p-start) at (0, 0.5) {$p$};
% \node (p-end)   at (0, -1.8) {};
% \node (q-start) at (2, 0.5) {$q$};
% \node (q-end)   at (2, -1.8) {};
% \node (r-start) at (4, 0.5) {$r$};
% \node (r-end)   at (4, -1.8) {};
% \draw (p-start) -- (p-end);
% \draw (q-start) -- (q-end);
% \draw (r-start) -- (r-end);

% % Messages
% \draw[thick,->] (2, 0) to node [above] {$A$} (0, -1.2) node [left] {$\mathcal{M}_p = \{A\}$};

% \draw[thick,->] (2, -0.1) to node [above] {$B$} (4, -1.3) node [right] {$\mathcal{M}_r = \{B\}$};

% \end{tikzpicture}

%% file: figs/trivial2.tikz
\begin{tikzpicture}

% Space between timelines
\def\width{1}
% Message delay
\def\delay{0.7}

% Timelimes
\node (p-start) at (0, 0.5) {$p$};
\node (p-end)   at (0, -3) {};
\node (q-start) at (\width, 0.5) {$q$};
\node (q-end)   at (\width, -3) {};
\node (r-start) at (\width*2, 0.5) {$r$};
\node (r-end)   at (\width*2, -3) {};
\draw (p-start) -- (p-end);
\draw (q-start) -- (q-end);
\draw (r-start) -- (r-end);

% Messages
\draw[thick,->] (\width, 0) to node [above,pos=0.4,sloped] {$\{A\}$} (0, -\delay) node [left] {$\mathcal{M}_p = \{A\}$};

\draw[thick,->] (\width, -0.1) to node [above,pos=0.4,sloped] {$\{B\}$} (\width*2, -\delay-0.1) node [right] {$\mathcal{M}_r = \{B\}$};

\draw[thick,->] (0, -1.5) to node [above,pos=0.2,sloped] {$\{A\}$} (\width*2, -1.5-\delay) node [right] {$\mathcal{M}_r' = \{A,B\}$};

\draw[thick,->] (\width*2, -1.5) to node [above,pos=0.2,sloped] {$\{B\}$} (0, -1.5-\delay) node [left] {$\mathcal{M}_p' = \{A,B\}$};

\end{tikzpicture}

% \begin{tikzpicture}
% % Timelimes
% \node (p-start) at (0, 0.5) {$p$};
% \node (p-end)   at (0, -3.4) {};
% \node (q-start) at (2, 0.5) {$q$};
% \node (q-end)   at (2, -3.4) {};
% \node (r-start) at (4, 0.5) {$r$};
% \node (r-end)   at (4, -3.4) {};
% \draw (p-start) -- (p-end);
% \draw (q-start) -- (q-end);
% \draw (r-start) -- (r-end);

% % Messages
% \draw[thick,->] (2, 0) to node [above] {$\{A\}$} (0, -1.2) node [left] {$\mathcal{M}_p = \{A\}$};

% \draw[thick,->] (2, -0.1) to node [above] {$\{B\}$} (4, -1.3) node [right] {$\mathcal{M}_r = \{B\}$};

% \draw[thick,->] (0, -1.7) to node [above,pos=0.25] {$\{A\}$} (4, -2.9) node [right] {$\mathcal{M}_r' = \{A,B\}$};

% \draw[thick,->] (4, -1.7) to node [above,pos=0.25] {$\{B\}$} (0, -2.9) node [left] {$\mathcal{M}_p' = \{A,B\}$};

% \end{tikzpicture}

%% file: figs/vectorclocks.tikz
\begin{tikzpicture}
% Message delay
\def\delay{0.6}

% Timelimes
\node (p-start) at (0, 0.5) {$p$};
\node (p-end)   at (0, -2.8) {};
\node (q-start) at (2.5, 0.5) {$q$};
\node (q-end)   at (2.5, -2.8) {};
\node (r-start) at (5, 0.5) {$r$};
\node (r-end)   at (5, -2.8) {};
\draw (p-start) -- (p-end);
\draw (q-start) -- (q-end);
\draw (r-start) -- (r-end);

% Messages
\draw[thick,->] (2.5, 0) to node [above,sloped] {$\{(0,1,0): A\}$} (0, -\delay) node [left] {$\mathit{vec} = (0,1,0),\; \mathcal{M}_p = \{A\}$};

\draw[thick,->] (2.5, -0.1) to node [above,sloped] {$\{(0,1,0): B\}$} (5, -0.1-\delay) node [right] {$\mathit{vec} = (0,1,0),\; \mathcal{M}_r = \{B\}$};

\draw[thick,->] (0, -1.2) to node [above,pos=0.25,sloped] {$(0,1,0)$} (5, -1.7-\delay) node [right] {$\mathit{vec} = (0,1,0),\; \mathcal{M}_r = \{B\}$};

\draw[thick,->] (5, -1.2) to node [above,pos=0.25,sloped] {$(0,1,0)$} (0, -1.7-\delay) node [left] {$\mathit{vec} = (0,1,0),\; \mathcal{M}_p = \{A\}$};

\end{tikzpicture}

%% file: figs/dag-before-p.tikz
\begin{tikzpicture}[node distance=0.9cm]

% nodes
\node (a) {$A$};
\node (b) [right of=a] {$B$};
\node (c) [above right of=b] {$C$};
\node (d) [right of=c] {$D$};
\node (e) [right of=d,draw,circle] {$E$};
\node (j) [below right of=b] {$J$};
\node (k) [right of=j] {$K$};
\node (l) [right of=k] {$L$};
\node (m) [right of=l,draw,circle] {$M$};

% arrows
\draw[<-] (a) -- (b);
\draw[<-] (b) -- (c);
\draw[<-] (c) -- (d);
\draw[<-] (d) -- (e);
\draw[<-] (j) -- (e);
\draw[<-] (b) -- (j);
\draw[<-] (j) -- (k);
\draw[<-] (k) -- (l);
\draw[<-] (l) -- (m);
\end{tikzpicture}

%% file: figs/dag-before-q.tikz
\begin{tikzpicture}[node distance=0.9cm]

% nodes
\node (a) {$A$};
\node (b) [right of=a] {$B$};
\node (f) [above right of=b] {$F$};
\node (g) [right of=f,circle,draw] {$G$};
\node (j) [below right of=b] {$J$};
\node (k) [right of=j,circle,draw] {$K$};

% arrows
\draw[<-] (a) -- (b);
\draw[<-] (b) -- (f);
\draw[<-] (f) -- (g);
\draw[<-] (b) -- (j);
\draw[<-] (j) -- (k);
\end{tikzpicture}

%% file: figs/dag-after.tikz
\begin{tikzpicture}[node distance=0.9cm]

% nodes
\node (a) {$A$};
\node (b) [right of=a] {$B$};
\node (c) [above right of=b] {$C$};
\node (d) [right of=c] {$D$};
\node (e) [right of=d,draw,circle] {$E$};
\node (j) [right of=b] {$J$};
\node (k) [right of=j] {$K$};
\node (l) [right of=k] {$L$};
\node (m) [right of=l,draw,circle] {$M$};
\node (f) [below right of=b] {$F$};
\node (g) [right of=f,draw,circle] {$G$};

% arrows
\draw[<-] (a) -- (b);
\draw[<-] (b) -- (c);
\draw[<-] (c) -- (d);
\draw[<-] (d) -- (e);
\draw[<-] (b) -- (j);
\draw[<-] (j) -- (e);
\draw[<-] (j) -- (k);
\draw[<-] (k) -- (l);
\draw[<-] (l) -- (m);
\draw[<-] (b) -- (f);
\draw[<-] (f) -- (g);
\end{tikzpicture}

%% file: figs/message-exchange.tikz
\begin{tikzpicture}
\newlength{\width}\setlength{\width}{2cm}
\newlength{\latency}\setlength{\latency}{0.9cm}
\newlength{\replydelay}\setlength{\replydelay}{0.3cm}
\tikzstyle{msg}=[thick,->]

% Timelimes
\node (p1-start) at (0, 0.5cm) {$p$};
\node (p2-start) at (\width, 0.5cm) {$q$};
\node (p1-end) at (0,-6.5cm) {};
\node (p2-end) at (\width,-6.5cm) {};
\draw (p1-start) -- (p1-end);
\draw (p2-start) -- (p2-end);

%\node (p-start) at (-0.5cm, \width) {$p$};
%\node (q-start) at (-0.5cm, 0) {$q$};
%\node (p-end) at (14cm, \distance) {};
%\node (q-end) at (14cm, 0) {};
%\draw (p-start) -- (p-end);
%\draw (q-start) -- (q-end);

% Messages
\draw[msg] (0,0) node[left] {$\langle\mathsf{heads}: \{H(E),H(M)\}\rangle$} -- (\width,\replydelay-\latency);
\draw[msg] (\width,0) node[right] {$\langle\mathsf{heads}: \{H(G),H(K)\}\rangle$} -- (0,\replydelay-\latency);

\draw[msg] (\width, -\latency) node[right] {$\langle\mathsf{needs}: \{H(E),H(M)\}\rangle$} -- (0,\replydelay-2.0\latency);
\draw[msg] (0, -\latency) node[left] {$\langle\mathsf{needs}: \{H(G)\}\rangle$} -- (\width,\replydelay-2.0\latency);

\draw[msg] (0, -2.0\latency) node[left] {$\langle\mathsf{msgs}: \{E,M\}\rangle$} -- (\width,\replydelay-3.0\latency);
\draw[msg] (\width, -2.0\latency) node[right] {$\langle\mathsf{msgs}: \{G\}\rangle$} -- (0,\replydelay-3.0\latency);

\draw[msg] (\width, -3.0\latency) node[right] {$\langle\mathsf{needs}: \{H(D),H(L)\}$} -- (0,\replydelay-4.0\latency);
\draw[msg] (0, -3.0\latency) node[left] {$\langle\mathsf{needs}: \{H(F)\}\rangle$} -- (\width,\replydelay-4.0\latency);

\draw[msg] (0, -4.0\latency) node[left] {$\langle\mathsf{msgs}: \{D,L\}\rangle$} -- (\width,\replydelay-5.0\latency);
\draw[msg] (\width, -4.0\latency) node[right] {$\langle\mathsf{msgs}: \{F\}\rangle$} -- (0,\replydelay-5.0\latency) node[left] {reconciliation complete};

\draw[msg] (\width, -5.0\latency) node[right] {$\langle\mathsf{needs}: \{H(C)\}$} -- (0,\replydelay-6.0\latency);

\draw[msg] (0, -6.0\latency) node[left] {$\langle\mathsf{msgs}: \{C\}\rangle$} -- (\width,\replydelay-7.0\latency) node[right] {reconciliation complete};

\end{tikzpicture}
\centering

%% file: byzantine-eventual.bbl
\begin{thebibliography}{80}

%%% the Web address.

\ifx \showCODEN    \undefined \def \showCODEN     #1{\unskip}     \fi
\ifx \showDOI      \undefined \def \showDOI       #1{#1}\fi
\ifx \showISBNx    \undefined \def \showISBNx     #1{\unskip}     \fi
\ifx \showISBNxiii \undefined \def \showISBNxiii  #1{\unskip}     \fi
\ifx \showISSN     \undefined \def \showISSN      #1{\unskip}     \fi
\ifx \showLCCN     \undefined \def \showLCCN      #1{\unskip}     \fi
\ifx \shownote     \undefined \def \shownote      #1{#1}          \fi
\ifx \showarticletitle \undefined \def \showarticletitle #1{#1}   \fi
\ifx \showURL      \undefined \def \showURL       {\relax}        \fi
% invisible to TeX
\providecommand\bibfield[2]{#2}
\providecommand\bibinfo[2]{#2}
\providecommand\natexlab[1]{#1}
\providecommand\showeprint[2][]{arXiv:#2}

\bibitem[\protect\citeauthoryear{Abd-El-Malek, Ganger, Goodson, Reiter, and
  Wylie}{Abd-El-Malek et~al\mbox{.}}{2005}]%
        {Abd:2005}
\bibfield{author}{\bibinfo{person}{Michael Abd-El-Malek},
  \bibinfo{person}{Gregory~R. Ganger}, \bibinfo{person}{Garth~R. Goodson},
  \bibinfo{person}{Michael~K. Reiter}, {and} \bibinfo{person}{Jay~J. Wylie}.}
  \bibinfo{year}{2005}\natexlab{}.
\newblock \showarticletitle{Fault-Scalable {Byzantine} Fault-Tolerant
  Services}.
\newblock \bibinfo{journal}{\emph{SIGOPS Operating Systems Review}}
  \bibinfo{volume}{39}, \bibinfo{number}{5} (\bibinfo{date}{Oct.}
  \bibinfo{year}{2005}), \bibinfo{pages}{59--74}.
\newblock
\showISSN{0163-5980}
\urldef\tempurl%
\url{https://doi.org/10.1145/1095809.1095817}
\showDOI{\tempurl}


\bibitem[\protect\citeauthoryear{Abraham, Devadas, Dolev, Nayak, and
  Ren}{Abraham et~al\mbox{.}}{2017}]%
        {Abraham:2017}
\bibfield{author}{\bibinfo{person}{Ittai Abraham}, \bibinfo{person}{Srinivas
  Devadas}, \bibinfo{person}{Danny Dolev}, \bibinfo{person}{Kartik Nayak},
  {and} \bibinfo{person}{Ling Ren}.} \bibinfo{year}{2017}\natexlab{}.
\newblock \showarticletitle{Efficient Synchronous Byzantine Consensus}.
\newblock \bibinfo{journal}{\emph{arXiv}} (\bibinfo{date}{Sept.}
  \bibinfo{year}{2017}).
\newblock
\showeprint{1704.02397}
\urldef\tempurl%
\url{https://arxiv.org/abs/1704.02397}
\showURL{%
\tempurl}


\bibitem[\protect\citeauthoryear{Adya, Liskov, and O'Neil}{Adya
  et~al\mbox{.}}{2000}]%
        {Adya:2000}
\bibfield{author}{\bibinfo{person}{Atul Adya}, \bibinfo{person}{Barbara
  Liskov}, {and} \bibinfo{person}{Patrick O'Neil}.}
  \bibinfo{year}{2000}\natexlab{}.
\newblock \showarticletitle{Generalized isolation level definitions}. In
  \bibinfo{booktitle}{\emph{16th International Conference on Data Engineering}}
  \emph{(\bibinfo{series}{ICDE 2000})}. \bibinfo{pages}{67--78}.
\newblock
\urldef\tempurl%
\url{https://doi.org/10.1109/icde.2000.839388}
\showDOI{\tempurl}


\bibitem[\protect\citeauthoryear{Akkoorath, Tomsic, Bravo, Li, Crain, Bieniusa,
  Pregui{\c c}a, and Shapiro}{Akkoorath et~al\mbox{.}}{2016}]%
        {Akkoorath2016Cure}
\bibfield{author}{\bibinfo{person}{Deepthi~Devaki Akkoorath},
  \bibinfo{person}{Alejandro~Z. Tomsic}, \bibinfo{person}{Manuel Bravo},
  \bibinfo{person}{Zhongmiao Li}, \bibinfo{person}{Tyler Crain},
  \bibinfo{person}{Annette Bieniusa}, \bibinfo{person}{Nuno Pregui{\c c}a},
  {and} \bibinfo{person}{Marc Shapiro}.} \bibinfo{year}{2016}\natexlab{}.
\newblock \showarticletitle{{Cure}: Strong Semantics Meets High Availability
  and Low Latency}. In \bibinfo{booktitle}{\emph{36th IEEE International
  Conference on Distributed Computing Systems}} \emph{(\bibinfo{series}{ICDCS
  2016})}. \bibinfo{publisher}{IEEE}, \bibinfo{pages}{405--414}.
\newblock
\urldef\tempurl%
\url{https://doi.org/10.1109/ICDCS.2016.98}
\showDOI{\tempurl}


\bibitem[\protect\citeauthoryear{Ameloot, Neven, and Van Den~Bussche}{Ameloot
  et~al\mbox{.}}{2013}]%
        {Ameloot:2013}
\bibfield{author}{\bibinfo{person}{Tom~J. Ameloot}, \bibinfo{person}{Frank
  Neven}, {and} \bibinfo{person}{Jan Van Den~Bussche}.}
  \bibinfo{year}{2013}\natexlab{}.
\newblock \showarticletitle{Relational Transducers for Declarative Networking}.
\newblock \bibinfo{journal}{\emph{J. ACM}} \bibinfo{volume}{60},
  \bibinfo{number}{2}, Article \bibinfo{articleno}{15} (\bibinfo{date}{May}
  \bibinfo{year}{2013}).
\newblock
\showISSN{0004-5411}
\urldef\tempurl%
\url{https://doi.org/10.1145/2450142.2450151}
\showDOI{\tempurl}


\bibitem[\protect\citeauthoryear{Bailis, Fekete, Franklin, Ghodsi, Hellerstein,
  and Stoica}{Bailis et~al\mbox{.}}{2014a}]%
        {Bailis:2014}
\bibfield{author}{\bibinfo{person}{Peter Bailis}, \bibinfo{person}{Alan
  Fekete}, \bibinfo{person}{Michael~J Franklin}, \bibinfo{person}{Ali Ghodsi},
  \bibinfo{person}{Joseph~M Hellerstein}, {and} \bibinfo{person}{Ion Stoica}.}
  \bibinfo{year}{2014}\natexlab{a}.
\newblock \showarticletitle{Coordination avoidance in database systems}.
\newblock \bibinfo{journal}{\emph{Proceedings of the VLDB Endowment}}
  \bibinfo{volume}{8}, \bibinfo{number}{3} (\bibinfo{date}{Nov.}
  \bibinfo{year}{2014}), \bibinfo{pages}{185--196}.
\newblock
\urldef\tempurl%
\url{https://doi.org/10.14778/2735508.2735509}
\showDOI{\tempurl}


\bibitem[\protect\citeauthoryear{Bailis, Fekete, Franklin, Ghodsi, Hellerstein,
  and Stoica}{Bailis et~al\mbox{.}}{2014b}]%
        {Bailis:2014ext}
\bibfield{author}{\bibinfo{person}{Peter Bailis}, \bibinfo{person}{Alan
  Fekete}, \bibinfo{person}{Michael~J Franklin}, \bibinfo{person}{Ali Ghodsi},
  \bibinfo{person}{Joseph~M Hellerstein}, {and} \bibinfo{person}{Ion Stoica}.}
  \bibinfo{year}{2014}\natexlab{b}.
\newblock \showarticletitle{Coordination Avoidance in Database Systems
  (Extended Version)}.
\newblock \bibinfo{journal}{\emph{arXiv}} (\bibinfo{date}{Oct.}
  \bibinfo{year}{2014}).
\newblock
\showeprint{1402.2237}
\urldef\tempurl%
\url{https://arxiv.org/abs/1402.2237}
\showURL{%
\tempurl}


\bibitem[\protect\citeauthoryear{Baird}{Baird}{2016}]%
        {Baird:2016tq}
\bibfield{author}{\bibinfo{person}{Leemon Baird}.}
  \bibinfo{year}{2016}\natexlab{}.
\newblock \bibinfo{booktitle}{\emph{The {Swirlds} hashgraph consensus
  algorithm: Fair, fast, {Byzantine} fault tolerance}}.
\newblock \bibinfo{type}{{T}echnical {R}eport} TR-2016-01.
  \bibinfo{institution}{Swirlds}.
\newblock
\urldef\tempurl%
\url{https://www.swirlds.com/downloads/SWIRLDS-TR-2016-01.pdf}
\showURL{%
\tempurl}


\bibitem[\protect\citeauthoryear{Bano, Sonnino, Al-Bassam, Azouvi, McCorry,
  Meiklejohn, and Danezis}{Bano et~al\mbox{.}}{2019}]%
        {Bano:2019}
\bibfield{author}{\bibinfo{person}{Shehar Bano}, \bibinfo{person}{Alberto
  Sonnino}, \bibinfo{person}{Mustafa Al-Bassam}, \bibinfo{person}{Sarah
  Azouvi}, \bibinfo{person}{Patrick McCorry}, \bibinfo{person}{Sarah
  Meiklejohn}, {and} \bibinfo{person}{George Danezis}.}
  \bibinfo{year}{2019}\natexlab{}.
\newblock \showarticletitle{{SoK}: Consensus in the Age of Blockchains}. In
  \bibinfo{booktitle}{\emph{1st ACM Conference on Advances in Financial
  Technologies}} \emph{(\bibinfo{series}{AFT 2019})}. \bibinfo{publisher}{ACM},
  \bibinfo{pages}{183--198}.
\newblock
\urldef\tempurl%
\url{https://doi.org/10.1145/3318041.3355458}
\showDOI{\tempurl}


\bibitem[\protect\citeauthoryear{Berenson, Bernstein, Gray, Melton, O'Neil, and
  O'Neil}{Berenson et~al\mbox{.}}{1995}]%
        {Berenson:1995}
\bibfield{author}{\bibinfo{person}{Hal Berenson}, \bibinfo{person}{Philip~A
  Bernstein}, \bibinfo{person}{Jim~N Gray}, \bibinfo{person}{Jim Melton},
  \bibinfo{person}{Elizabeth O'Neil}, {and} \bibinfo{person}{Patrick O'Neil}.}
  \bibinfo{year}{1995}\natexlab{}.
\newblock \showarticletitle{A critique of {ANSI SQL} isolation levels}.
\newblock \bibinfo{journal}{\emph{ACM SIGMOD Record}} \bibinfo{volume}{24},
  \bibinfo{number}{2} (\bibinfo{date}{May} \bibinfo{year}{1995}),
  \bibinfo{pages}{1--10}.
\newblock
\urldef\tempurl%
\url{https://doi.org/10.1145/568271.223785}
\showDOI{\tempurl}


\bibitem[\protect\citeauthoryear{Bessani, Sousa, and Alchieri}{Bessani
  et~al\mbox{.}}{2014}]%
        {Bessani:2014}
\bibfield{author}{\bibinfo{person}{Alysson Bessani}, \bibinfo{person}{Jo\~{a}o
  Sousa}, {and} \bibinfo{person}{Eduardo E.~P. Alchieri}.}
  \bibinfo{year}{2014}\natexlab{}.
\newblock \showarticletitle{State Machine Replication for the Masses with
  {BFT-SMART}}. In \bibinfo{booktitle}{\emph{44th Annual IEEE/IFIP
  International Conference on Dependable Systems and Networks}}
  \emph{(\bibinfo{series}{DSN 2014})}. \bibinfo{publisher}{IEEE},
  \bibinfo{pages}{355--362}.
\newblock
\urldef\tempurl%
\url{https://doi.org/10.1109/DSN.2014.43}
\showDOI{\tempurl}


\bibitem[\protect\citeauthoryear{Birman, Schiper, and Stephenson}{Birman
  et~al\mbox{.}}{1991}]%
        {Birman:1991el}
\bibfield{author}{\bibinfo{person}{Kenneth Birman}, \bibinfo{person}{Andr{\'e}
  Schiper}, {and} \bibinfo{person}{Pat Stephenson}.}
  \bibinfo{year}{1991}\natexlab{}.
\newblock \showarticletitle{Lightweight causal and atomic group multicast}.
\newblock \bibinfo{journal}{\emph{ACM Transactions on Computer Systems}}
  \bibinfo{volume}{9}, \bibinfo{number}{3} (\bibinfo{date}{Aug.}
  \bibinfo{year}{1991}), \bibinfo{pages}{272--314}.
\newblock
\urldef\tempurl%
\url{https://doi.org/10.1145/128738.128742}
\showDOI{\tempurl}


\bibitem[\protect\citeauthoryear{Bloom}{Bloom}{1970}]%
        {Bloom:1970}
\bibfield{author}{\bibinfo{person}{Burton~H. Bloom}.}
  \bibinfo{year}{1970}\natexlab{}.
\newblock \showarticletitle{Space/Time Trade-Offs in Hash Coding with Allowable
  Errors}.
\newblock \bibinfo{journal}{\emph{Commun. ACM}} \bibinfo{volume}{13},
  \bibinfo{number}{7} (\bibinfo{date}{July} \bibinfo{year}{1970}),
  \bibinfo{pages}{422--426}.
\newblock
\showISSN{0001-0782}
\urldef\tempurl%
\url{https://doi.org/10.1145/362686.362692}
\showDOI{\tempurl}


\bibitem[\protect\citeauthoryear{Bose, Guo, Kranakis, Maheshwari, Morin,
  Morrison, Smid, and Tang}{Bose et~al\mbox{.}}{2008}]%
        {Bose:2008}
\bibfield{author}{\bibinfo{person}{Prosenjit Bose}, \bibinfo{person}{Hua Guo},
  \bibinfo{person}{Evangelos Kranakis}, \bibinfo{person}{Anil Maheshwari},
  \bibinfo{person}{Pat Morin}, \bibinfo{person}{Jason Morrison},
  \bibinfo{person}{Michiel Smid}, {and} \bibinfo{person}{Yihui Tang}.}
  \bibinfo{year}{2008}\natexlab{}.
\newblock \showarticletitle{On the False-Positive Rate of Bloom Filters}.
\newblock \bibinfo{journal}{\emph{Inf. Process. Lett.}} \bibinfo{volume}{108},
  \bibinfo{number}{4} (\bibinfo{date}{Oct.} \bibinfo{year}{2008}),
  \bibinfo{pages}{210–213}.
\newblock
\showISSN{0020-0190}
\urldef\tempurl%
\url{https://doi.org/10.1016/j.ipl.2008.05.018}
\showDOI{\tempurl}


\bibitem[\protect\citeauthoryear{Cachin, Guerraoui, and Rodrigues}{Cachin
  et~al\mbox{.}}{2011}]%
        {Cachin:2011wt}
\bibfield{author}{\bibinfo{person}{Christian Cachin}, \bibinfo{person}{Rachid
  Guerraoui}, {and} \bibinfo{person}{Luís Rodrigues}.}
  \bibinfo{year}{2011}\natexlab{}.
\newblock \bibinfo{booktitle}{\emph{Introduction to Reliable and Secure
  Distributed Programming} (\bibinfo{edition}{second} ed.)}.
\newblock \bibinfo{publisher}{Springer}.
\newblock
\showISBNx{9783642152597}


\bibitem[\protect\citeauthoryear{Cachin, Kursawe, Petzold, and Shoup}{Cachin
  et~al\mbox{.}}{2001}]%
        {Cachin:2001cj}
\bibfield{author}{\bibinfo{person}{Christian Cachin}, \bibinfo{person}{Klaus
  Kursawe}, \bibinfo{person}{Frank Petzold}, {and} \bibinfo{person}{Victor
  Shoup}.} \bibinfo{year}{2001}\natexlab{}.
\newblock \showarticletitle{Secure and Efficient Asynchronous Broadcast
  Protocols}. In \bibinfo{booktitle}{\emph{21st Annual International Cryptology
  Conference}} \emph{(\bibinfo{series}{CRYPTO 2001})}.
  \bibinfo{publisher}{Springer}, \bibinfo{pages}{524--541}.
\newblock
\urldef\tempurl%
\url{https://doi.org/10.1007/3-540-44647-8_31}
\showDOI{\tempurl}


\bibitem[\protect\citeauthoryear{Castro and Liskov}{Castro and Liskov}{1999}]%
        {Castro:1999}
\bibfield{author}{\bibinfo{person}{Miguel Castro} {and}
  \bibinfo{person}{Barbara Liskov}.} \bibinfo{year}{1999}\natexlab{}.
\newblock \showarticletitle{Practical {Byzantine} Fault Tolerance}. In
  \bibinfo{booktitle}{\emph{3rd Symposium on Operating Systems Design and
  Implementation}} \emph{(\bibinfo{series}{OSDI 1999})}.
  \bibinfo{publisher}{USENIX Association}, \bibinfo{pages}{173–186}.
\newblock


\bibitem[\protect\citeauthoryear{Chai and Zhao}{Chai and Zhao}{2014}]%
        {Chai:2014}
\bibfield{author}{\bibinfo{person}{Hua Chai} {and} \bibinfo{person}{Wenbing
  Zhao}.} \bibinfo{year}{2014}\natexlab{}.
\newblock \showarticletitle{Byzantine Fault Tolerance for Services with
  Commutative Operations}. In \bibinfo{booktitle}{\emph{2014 IEEE International
  Conference on Services Computing}} \emph{(\bibinfo{series}{SCC 2014})}.
  \bibinfo{publisher}{IEEE}, \bibinfo{pages}{219--226}.
\newblock
\urldef\tempurl%
\url{https://doi.org/10.1109/SCC.2014.37}
\showDOI{\tempurl}


\bibitem[\protect\citeauthoryear{Christensen, Roginsky, and Jimeno}{Christensen
  et~al\mbox{.}}{2010}]%
        {Christensen:2010}
\bibfield{author}{\bibinfo{person}{Ken Christensen}, \bibinfo{person}{Allen
  Roginsky}, {and} \bibinfo{person}{Miguel Jimeno}.}
  \bibinfo{year}{2010}\natexlab{}.
\newblock \showarticletitle{A New Analysis of the False Positive Rate of a
  Bloom Filter}.
\newblock \bibinfo{journal}{\emph{Inf. Process. Lett.}} \bibinfo{volume}{110},
  \bibinfo{number}{21} (\bibinfo{date}{Oct.} \bibinfo{year}{2010}),
  \bibinfo{pages}{944–949}.
\newblock
\showISSN{0020-0190}
\urldef\tempurl%
\url{https://doi.org/10.1016/j.ipl.2010.07.024}
\showDOI{\tempurl}


\bibitem[\protect\citeauthoryear{Clement, Kapritsos, Lee, Wang, Alvisi, Dahlin,
  and Riche}{Clement et~al\mbox{.}}{2009}]%
        {Clement:2009}
\bibfield{author}{\bibinfo{person}{Allen Clement}, \bibinfo{person}{Manos
  Kapritsos}, \bibinfo{person}{Sangmin Lee}, \bibinfo{person}{Yang Wang},
  \bibinfo{person}{Lorenzo Alvisi}, \bibinfo{person}{Mike Dahlin}, {and}
  \bibinfo{person}{Taylor Riche}.} \bibinfo{year}{2009}\natexlab{}.
\newblock \showarticletitle{Upright Cluster Services}. In
  \bibinfo{booktitle}{\emph{22nd ACM SIGOPS Symposium on Operating Systems
  Principles}} \emph{(\bibinfo{series}{SOSP 2009})}. \bibinfo{publisher}{ACM},
  \bibinfo{pages}{277--290}.
\newblock
\urldef\tempurl%
\url{https://doi.org/10.1145/1629575.1629602}
\showDOI{\tempurl}


\bibitem[\protect\citeauthoryear{Conway, Marczak, Alvaro, Hellerstein, and
  Maier}{Conway et~al\mbox{.}}{2012}]%
        {Conway:2012}
\bibfield{author}{\bibinfo{person}{Neil Conway}, \bibinfo{person}{William~R.
  Marczak}, \bibinfo{person}{Peter Alvaro}, \bibinfo{person}{Joseph~M.
  Hellerstein}, {and} \bibinfo{person}{David Maier}.}
  \bibinfo{year}{2012}\natexlab{}.
\newblock \showarticletitle{Logic and Lattices for Distributed Programming}. In
  \bibinfo{booktitle}{\emph{3rd ACM Symposium on Cloud Computing}}
  \emph{(\bibinfo{series}{SoCC 2012})}. \bibinfo{pages}{1--14}.
\newblock
\urldef\tempurl%
\url{https://doi.org/10.1145/2391229.2391230}
\showDOI{\tempurl}


\bibitem[\protect\citeauthoryear{D{\'e}fago, Schiper, and Urb{\'a}n}{D{\'e}fago
  et~al\mbox{.}}{2004}]%
        {Defago:2004ji}
\bibfield{author}{\bibinfo{person}{Xavier D{\'e}fago},
  \bibinfo{person}{Andr{\'e} Schiper}, {and} \bibinfo{person}{P{\'e}ter
  Urb{\'a}n}.} \bibinfo{year}{2004}\natexlab{}.
\newblock \showarticletitle{Total order broadcast and multicast algorithms:
  Taxonomy and survey}.
\newblock \bibinfo{journal}{\emph{Comput. Surveys}} \bibinfo{volume}{36},
  \bibinfo{number}{4} (\bibinfo{date}{Dec.} \bibinfo{year}{2004}),
  \bibinfo{pages}{372--421}.
\newblock
\urldef\tempurl%
\url{https://doi.org/10.1145/1041680.1041682}
\showDOI{\tempurl}


\bibitem[\protect\citeauthoryear{{Di Luna}, Anceaume, and Querzoni}{{Di Luna}
  et~al\mbox{.}}{2020}]%
        {DiLuna:2020}
\bibfield{author}{\bibinfo{person}{Giuseppe~Antonio {Di Luna}},
  \bibinfo{person}{Emmanuelle Anceaume}, {and} \bibinfo{person}{Leonardo
  Querzoni}.} \bibinfo{year}{2020}\natexlab{}.
\newblock \showarticletitle{Byzantine Generalized Lattice Agreement}. In
  \bibinfo{booktitle}{\emph{IEEE International Parallel and Distributed
  Processing Symposium}} \emph{(\bibinfo{series}{IPDPS 2020})}.
  \bibinfo{publisher}{IEEE}, \bibinfo{pages}{674--683}.
\newblock
\urldef\tempurl%
\url{https://doi.org/10.1109/IPDPS47924.2020.00075}
\showDOI{\tempurl}


\bibitem[\protect\citeauthoryear{Douceur}{Douceur}{2002}]%
        {Douceur:2002}
\bibfield{author}{\bibinfo{person}{John~R. Douceur}.}
  \bibinfo{year}{2002}\natexlab{}.
\newblock \showarticletitle{The {Sybil} Attack}. In
  \bibinfo{booktitle}{\emph{International Workshop on Peer-to-Peer Systems}}
  \emph{(\bibinfo{series}{IPTPS 2002})}. \bibinfo{publisher}{Springer},
  \bibinfo{pages}{251--260}.
\newblock
\urldef\tempurl%
\url{https://doi.org/10.1007/3-540-45748-8_24}
\showDOI{\tempurl}


\bibitem[\protect\citeauthoryear{Drabkin, Friedman, and Segal}{Drabkin
  et~al\mbox{.}}{2005}]%
        {Drabkin:2005}
\bibfield{author}{\bibinfo{person}{Vadim Drabkin}, \bibinfo{person}{Roy
  Friedman}, {and} \bibinfo{person}{Marc Segal}.}
  \bibinfo{year}{2005}\natexlab{}.
\newblock \showarticletitle{Efficient Byzantine Broadcast in Wireless Ad-Hoc
  Networks}. In \bibinfo{booktitle}{\emph{Proceedings of the 2005 International
  Conference on Dependable Systems and Networks}} \emph{(\bibinfo{series}{DSN
  '05})}. \bibinfo{publisher}{IEEE Computer Society}, \bibinfo{address}{USA},
  \bibinfo{pages}{160–169}.
\newblock
\showISBNx{0769522823}
\urldef\tempurl%
\url{https://doi.org/10.1109/DSN.2005.42}
\showDOI{\tempurl}


\bibitem[\protect\citeauthoryear{Duan, Reiter, and Zhang}{Duan
  et~al\mbox{.}}{2017}]%
        {Duan:2017}
\bibfield{author}{\bibinfo{person}{Sisi Duan}, \bibinfo{person}{Michael~K.
  Reiter}, {and} \bibinfo{person}{Haibin Zhang}.}
  \bibinfo{year}{2017}\natexlab{}.
\newblock \showarticletitle{Secure Causal Atomic Broadcast, Revisited}. In
  \bibinfo{booktitle}{\emph{47th Annual IEEE/IFIP International Conference on
  Dependable Systems and Networks}} \emph{(\bibinfo{series}{DSN 2017})}.
  \bibinfo{publisher}{IEEE}, \bibinfo{pages}{61--72}.
\newblock
\urldef\tempurl%
\url{https://doi.org/10.1109/DSN.2017.64}
\showDOI{\tempurl}


\bibitem[\protect\citeauthoryear{Dwork, Lynch, and Stockmeyer}{Dwork
  et~al\mbox{.}}{1988}]%
        {Dwork:1988}
\bibfield{author}{\bibinfo{person}{Cynthia Dwork}, \bibinfo{person}{Nancy
  Lynch}, {and} \bibinfo{person}{Larry Stockmeyer}.}
  \bibinfo{year}{1988}\natexlab{}.
\newblock \showarticletitle{Consensus in the Presence of Partial Synchrony}.
\newblock \bibinfo{journal}{\emph{J. ACM}} \bibinfo{volume}{35},
  \bibinfo{number}{2} (\bibinfo{date}{April} \bibinfo{year}{1988}),
  \bibinfo{pages}{288--323}.
\newblock
\urldef\tempurl%
\url{https://doi.org/10.1145/42282.42283}
\showDOI{\tempurl}


\bibitem[\protect\citeauthoryear{Eppstein, Goodrich, Uyeda, and
  Varghese}{Eppstein et~al\mbox{.}}{2011}]%
        {Eppstein:2011}
\bibfield{author}{\bibinfo{person}{David Eppstein}, \bibinfo{person}{Michael~T.
  Goodrich}, \bibinfo{person}{Frank Uyeda}, {and} \bibinfo{person}{George
  Varghese}.} \bibinfo{year}{2011}\natexlab{}.
\newblock \showarticletitle{What's the Difference? {Efficient} Set
  Reconciliation without Prior Context}. In \bibinfo{booktitle}{\emph{ACM
  SIGCOMM 2011 Conference}}. \bibinfo{publisher}{ACM},
  \bibinfo{pages}{218--229}.
\newblock
\urldef\tempurl%
\url{https://doi.org/10.1145/2018436.2018462}
\showDOI{\tempurl}


\bibitem[\protect\citeauthoryear{Feldman, Zeller, Freedman, and Felten}{Feldman
  et~al\mbox{.}}{2010}]%
        {Feldman:2010wl}
\bibfield{author}{\bibinfo{person}{Ariel~J. Feldman},
  \bibinfo{person}{William~P. Zeller}, \bibinfo{person}{Michael~J. Freedman},
  {and} \bibinfo{person}{Edward~W. Felten}.} \bibinfo{year}{2010}\natexlab{}.
\newblock \showarticletitle{{SPORC}: Group Collaboration using Untrusted Cloud
  Resources}. In \bibinfo{booktitle}{\emph{9th USENIX Symposium on Operating
  Systems Design and Implementation}} \emph{(\bibinfo{series}{OSDI 2010})}.
  \bibinfo{publisher}{USENIX Association}.
\newblock


\bibitem[\protect\citeauthoryear{Gomes, Kleppmann, Mulligan, and
  Beresford}{Gomes et~al\mbox{.}}{2017}]%
        {Gomes:2017gy}
\bibfield{author}{\bibinfo{person}{Victor~B.F. Gomes}, \bibinfo{person}{Martin
  Kleppmann}, \bibinfo{person}{Dominic~P. Mulligan}, {and}
  \bibinfo{person}{Alastair~R. Beresford}.} \bibinfo{year}{2017}\natexlab{}.
\newblock \showarticletitle{Verifying strong eventual consistency in
  distributed systems}.
\newblock \bibinfo{journal}{\emph{Proceedings of the ACM on Programming
  Languages}} \bibinfo{volume}{1}, \bibinfo{number}{OOPSLA}
  (\bibinfo{date}{Oct.} \bibinfo{year}{2017}).
\newblock
\urldef\tempurl%
\url{https://doi.org/10.1145/3133933}
\showDOI{\tempurl}


\bibitem[\protect\citeauthoryear{Goodrich and Mitzenmacher}{Goodrich and
  Mitzenmacher}{2011}]%
        {Goodrich:2011}
\bibfield{author}{\bibinfo{person}{Michael~T. Goodrich} {and}
  \bibinfo{person}{Michael Mitzenmacher}.} \bibinfo{year}{2011}\natexlab{}.
\newblock \bibinfo{title}{Invertible Bloom Lookup Tables}.
\newblock
\newblock
\showeprint[arxiv]{1101.2245}~[cs.DS]
\urldef\tempurl%
\url{https://arxiv.org/abs/1101.2245}
\showURL{%
\tempurl}


\bibitem[\protect\citeauthoryear{Hellerstein}{Hellerstein}{2010}]%
        {Hellerstein:2010}
\bibfield{author}{\bibinfo{person}{Joseph~M. Hellerstein}.}
  \bibinfo{year}{2010}\natexlab{}.
\newblock \showarticletitle{The declarative imperative}.
\newblock \bibinfo{journal}{\emph{ACM SIGMOD Record}} \bibinfo{volume}{39},
  \bibinfo{number}{1} (\bibinfo{date}{Sept.} \bibinfo{year}{2010}),
  \bibinfo{pages}{5--19}.
\newblock
\urldef\tempurl%
\url{https://doi.org/10.1145/1860702.1860704}
\showDOI{\tempurl}


\bibitem[\protect\citeauthoryear{Kang, Wilensky, and Kubiatowicz}{Kang
  et~al\mbox{.}}{2003}]%
        {Kang:2003}
\bibfield{author}{\bibinfo{person}{Brent~Byunghoon Kang},
  \bibinfo{person}{Robert Wilensky}, {and} \bibinfo{person}{John Kubiatowicz}.}
  \bibinfo{year}{2003}\natexlab{}.
\newblock \showarticletitle{The hash history approach for reconciling mutual
  inconsistency}. In \bibinfo{booktitle}{\emph{23rd International Conference on
  Distributed Computing Systems}} \emph{(\bibinfo{series}{ICDCS 2003})}.
  \bibinfo{publisher}{IEEE}, \bibinfo{pages}{670--677}.
\newblock
\urldef\tempurl%
\url{https://doi.org/10.1109/ICDCS.2003.1203518}
\showDOI{\tempurl}


\bibitem[\protect\citeauthoryear{Kleppmann and Beresford}{Kleppmann and
  Beresford}{2017}]%
        {Kleppmann:2017}
\bibfield{author}{\bibinfo{person}{Martin Kleppmann} {and}
  \bibinfo{person}{Alastair~R Beresford}.} \bibinfo{year}{2017}\natexlab{}.
\newblock \showarticletitle{A Conflict-Free Replicated {JSON} Datatype}.
\newblock \bibinfo{journal}{\emph{IEEE Transactions on Parallel and Distributed
  Systems}} \bibinfo{volume}{28}, \bibinfo{number}{10} (\bibinfo{date}{April}
  \bibinfo{year}{2017}), \bibinfo{pages}{2733--2746}.
\newblock
\urldef\tempurl%
\url{https://doi.org/10.1109/TPDS.2017.2697382}
\showDOI{\tempurl}


\bibitem[\protect\citeauthoryear{Kleppmann, Wiggins, van Hardenberg, and
  McGranaghan}{Kleppmann et~al\mbox{.}}{2019}]%
        {Kleppmann2019localfirst}
\bibfield{author}{\bibinfo{person}{Martin Kleppmann}, \bibinfo{person}{Adam
  Wiggins}, \bibinfo{person}{Peter van Hardenberg}, {and} \bibinfo{person}{Mark
  McGranaghan}.} \bibinfo{year}{2019}\natexlab{}.
\newblock \showarticletitle{Local-First Software: You own your data, in spite
  of the cloud}. In \bibinfo{booktitle}{\emph{ACM SIGPLAN International
  Symposium on New Ideas, New Paradigms, and Reflections on Programming and
  Software}} \emph{(\bibinfo{series}{Onward! 2019})}. \bibinfo{publisher}{ACM},
  \bibinfo{pages}{154--178}.
\newblock
\urldef\tempurl%
\url{https://doi.org/10.1145/3359591.3359737}
\showDOI{\tempurl}


\bibitem[\protect\citeauthoryear{Kotla, Alvisi, Dahlin, Clement, and
  Wong}{Kotla et~al\mbox{.}}{2007}]%
        {Kotla:2007}
\bibfield{author}{\bibinfo{person}{Ramakrishna Kotla}, \bibinfo{person}{Lorenzo
  Alvisi}, \bibinfo{person}{Mike Dahlin}, \bibinfo{person}{Allen Clement},
  {and} \bibinfo{person}{Edmund Wong}.} \bibinfo{year}{2007}\natexlab{}.
\newblock \showarticletitle{Zyzzyva: Speculative {Byzantine} Fault Tolerance}.
\newblock \bibinfo{journal}{\emph{SIGOPS Operating Systems Review}}
  \bibinfo{volume}{41}, \bibinfo{number}{6} (\bibinfo{date}{Oct.}
  \bibinfo{year}{2007}), \bibinfo{pages}{45--58}.
\newblock
\showISSN{0163-5980}
\urldef\tempurl%
\url{https://doi.org/10.1145/1323293.1294267}
\showDOI{\tempurl}


\bibitem[\protect\citeauthoryear{Lamport, Shostak, and Pease}{Lamport
  et~al\mbox{.}}{1982}]%
        {Lamport:1982}
\bibfield{author}{\bibinfo{person}{Leslie Lamport}, \bibinfo{person}{Robert
  Shostak}, {and} \bibinfo{person}{Marshall Pease}.}
  \bibinfo{year}{1982}\natexlab{}.
\newblock \showarticletitle{The {Byzantine} Generals Problem}.
\newblock \bibinfo{journal}{\emph{ACM Transactions on Programming Languages and
  Systems}} \bibinfo{volume}{4}, \bibinfo{number}{3} (\bibinfo{date}{July}
  \bibinfo{year}{1982}), \bibinfo{pages}{382–401}.
\newblock
\showISSN{0164-0925}
\urldef\tempurl%
\url{https://doi.org/10.1145/357172.357176}
\showDOI{\tempurl}


\bibitem[\protect\citeauthoryear{Leitão, Pereira, and Rodrigues}{Leitão
  et~al\mbox{.}}{2009}]%
        {Leitao:2009fi}
\bibfield{author}{\bibinfo{person}{João Leitão}, \bibinfo{person}{José
  Pereira}, {and} \bibinfo{person}{Luís Rodrigues}.}
  \bibinfo{year}{2009}\natexlab{}.
\newblock \showarticletitle{Gossip-Based Broadcast}.
\newblock In \bibinfo{booktitle}{\emph{Handbook of Peer-to-Peer Networking}}.
  \bibinfo{publisher}{Springer}, \bibinfo{pages}{831--860}.
\newblock
\showISBNx{978-0-387-09750-3}
\urldef\tempurl%
\url{https://doi.org/10.1007/978-0-387-09751-0_29}
\showDOI{\tempurl}


\bibitem[\protect\citeauthoryear{Li and Mazi{\`e}res}{Li and
  Mazi{\`e}res}{2007}]%
        {Li:2007}
\bibfield{author}{\bibinfo{person}{Jinyuan Li} {and} \bibinfo{person}{David
  Mazi{\`e}res}.} \bibinfo{year}{2007}\natexlab{}.
\newblock \showarticletitle{Beyond One-third Faulty Replicas in {Byzantine}
  Fault Tolerant Systems}. In \bibinfo{booktitle}{\emph{4th USENIX Symposium on
  Networked Systems Design {\&} Implementation}} \emph{(\bibinfo{series}{NSDI
  2007})}. \bibinfo{publisher}{USENIX}, \bibinfo{pages}{131--144}.
\newblock


\bibitem[\protect\citeauthoryear{Linde, Leitão, and Preguiça}{Linde
  et~al\mbox{.}}{2020}]%
        {vanderLinde:2020}
\bibfield{author}{\bibinfo{person}{Albert van~der Linde},
  \bibinfo{person}{João Leitão}, {and} \bibinfo{person}{Nuno Preguiça}.}
  \bibinfo{year}{2020}\natexlab{}.
\newblock \showarticletitle{Practical Client-side Replication: Weak Consistency
  Semantics for Insecure Settings}.
\newblock \bibinfo{journal}{\emph{Proceedings of the VLDB Endowment}}
  \bibinfo{volume}{13}, \bibinfo{number}{11} (\bibinfo{date}{July}
  \bibinfo{year}{2020}), \bibinfo{pages}{2590--2605}.
\newblock
\urldef\tempurl%
\url{https://doi.org/10.14778/3407790.3407847}
\showDOI{\tempurl}


\bibitem[\protect\citeauthoryear{Lloyd, Freedman, Kaminsky, and Andersen}{Lloyd
  et~al\mbox{.}}{2011}]%
        {Lloyd:2011}
\bibfield{author}{\bibinfo{person}{Wyatt Lloyd}, \bibinfo{person}{Michael~J.
  Freedman}, \bibinfo{person}{Michael Kaminsky}, {and}
  \bibinfo{person}{David~G. Andersen}.} \bibinfo{year}{2011}\natexlab{}.
\newblock \showarticletitle{Don't Settle for Eventual: Scalable Causal
  Consistency for Wide-Area Storage with {COPS}}. In
  \bibinfo{booktitle}{\emph{23rd ACM Symposium on Operating Systems
  Principles}} \emph{(\bibinfo{series}{SOSP 2011})}. \bibinfo{publisher}{ACM},
  \bibinfo{pages}{401--416}.
\newblock
\urldef\tempurl%
\url{https://doi.org/10.1145/2043556.2043593}
\showDOI{\tempurl}


\bibitem[\protect\citeauthoryear{Lv, He, Cheng, and Wu}{Lv
  et~al\mbox{.}}{2018}]%
        {Lv:2018ie}
\bibfield{author}{\bibinfo{person}{Xiao Lv}, \bibinfo{person}{Fazhi He},
  \bibinfo{person}{Yuan Cheng}, {and} \bibinfo{person}{Yiqi Wu}.}
  \bibinfo{year}{2018}\natexlab{}.
\newblock \showarticletitle{A novel {CRDT}-based synchronization method for
  real-time collaborative {CAD} systems}.
\newblock \bibinfo{journal}{\emph{Advanced Engineering Informatics}}
  \bibinfo{volume}{38} (\bibinfo{date}{Aug.} \bibinfo{year}{2018}),
  \bibinfo{pages}{381--391}.
\newblock
\urldef\tempurl%
\url{https://doi.org/10.1016/j.aei.2018.08.008}
\showDOI{\tempurl}


\bibitem[\protect\citeauthoryear{Mahajan}{Mahajan}{2012}]%
        {Mahajan:2012}
\bibfield{author}{\bibinfo{person}{Prince Mahajan}.}
  \bibinfo{year}{2012}\natexlab{}.
\newblock \emph{\bibinfo{title}{Highly Available Storage with Minimal Trust}}.
\newblock \bibinfo{thesistype}{Ph.D. Dissertation}. \bibinfo{school}{University
  of Texas at Austin}.
\newblock
\urldef\tempurl%
\url{https://repositories.lib.utexas.edu/handle/2152/16320}
\showURL{%
\tempurl}


\bibitem[\protect\citeauthoryear{Mahajan, Alvisi, and Dahlin}{Mahajan
  et~al\mbox{.}}{2011a}]%
        {Mahajan:2011cac}
\bibfield{author}{\bibinfo{person}{Prince Mahajan}, \bibinfo{person}{Lorenzo
  Alvisi}, {and} \bibinfo{person}{Mike Dahlin}.}
  \bibinfo{year}{2011}\natexlab{a}.
\newblock \bibinfo{booktitle}{\emph{Consistency, Availability, and
  Convergence}}.
\newblock \bibinfo{type}{{T}echnical {R}eport} UTCS TR-11-22.
  \bibinfo{institution}{University of Texas at Austin}.
\newblock
\urldef\tempurl%
\url{https://www.cs.cornell.edu/lorenzo/papers/cac-tr.pdf}
\showURL{%
\tempurl}


\bibitem[\protect\citeauthoryear{Mahajan, Setty, Lee, Clement, Alvisi, Dahlin,
  and Walfish}{Mahajan et~al\mbox{.}}{2010}]%
        {Mahajan:2010}
\bibfield{author}{\bibinfo{person}{Prince Mahajan}, \bibinfo{person}{Srinath
  Setty}, \bibinfo{person}{Sangmin Lee}, \bibinfo{person}{Allen Clement},
  \bibinfo{person}{Lorenzo Alvisi}, \bibinfo{person}{Mike Dahlin}, {and}
  \bibinfo{person}{Michael Walfish}.} \bibinfo{year}{2010}\natexlab{}.
\newblock \showarticletitle{{Depot}: Cloud storage with minimal trust}. In
  \bibinfo{booktitle}{\emph{9th USENIX conference on Operating Systems Design
  and Implementation}} \emph{(\bibinfo{series}{OSDI 2010})}.
\newblock


\bibitem[\protect\citeauthoryear{Mahajan, Setty, Lee, Clement, Alvisi, Dahlin,
  and Walfish}{Mahajan et~al\mbox{.}}{2011b}]%
        {Mahajan:2011}
\bibfield{author}{\bibinfo{person}{Prince Mahajan}, \bibinfo{person}{Srinath
  Setty}, \bibinfo{person}{Sangmin Lee}, \bibinfo{person}{Allen Clement},
  \bibinfo{person}{Lorenzo Alvisi}, \bibinfo{person}{Mike Dahlin}, {and}
  \bibinfo{person}{Michael Walfish}.} \bibinfo{year}{2011}\natexlab{b}.
\newblock \showarticletitle{{Depot}: Cloud Storage with Minimal Trust}.
\newblock \bibinfo{journal}{\emph{ACM Transactions on Computer Systems}}
  \bibinfo{volume}{29}, \bibinfo{number}{4}, Article \bibinfo{articleno}{12}
  (\bibinfo{date}{Dec.} \bibinfo{year}{2011}).
\newblock
\urldef\tempurl%
\url{https://doi.org/10.1145/2063509.2063512}
\showDOI{\tempurl}


\bibitem[\protect\citeauthoryear{Malkhi, Merritt, and Rodeh}{Malkhi
  et~al\mbox{.}}{2000}]%
        {Malkhi:2000}
\bibfield{author}{\bibinfo{person}{Dahlia Malkhi}, \bibinfo{person}{Michael
  Merritt}, {and} \bibinfo{person}{Ohad Rodeh}.}
  \bibinfo{year}{2000}\natexlab{}.
\newblock \showarticletitle{Secure Reliable Multicast Protocols in a WAN}.
\newblock \bibinfo{journal}{\emph{Distrib. Comput.}} \bibinfo{volume}{13},
  \bibinfo{number}{1} (\bibinfo{date}{Jan.} \bibinfo{year}{2000}),
  \bibinfo{pages}{19–28}.
\newblock
\showISSN{0178-2770}
\urldef\tempurl%
\url{https://doi.org/10.1007/s004460050002}
\showDOI{\tempurl}


\bibitem[\protect\citeauthoryear{Malkhi and Reiter}{Malkhi and Reiter}{1998}]%
        {Malkhi:1998}
\bibfield{author}{\bibinfo{person}{Dahlia Malkhi} {and}
  \bibinfo{person}{Michael Reiter}.} \bibinfo{year}{1998}\natexlab{}.
\newblock \showarticletitle{Byzantine Quorum Systems}.
\newblock \bibinfo{journal}{\emph{Distrib. Comput.}} \bibinfo{volume}{11},
  \bibinfo{number}{4} (\bibinfo{date}{Oct.} \bibinfo{year}{1998}),
  \bibinfo{pages}{203–213}.
\newblock
\showISSN{0178-2770}
\urldef\tempurl%
\url{https://doi.org/10.1007/s004460050050}
\showDOI{\tempurl}


\bibitem[\protect\citeauthoryear{Malki and Reiter}{Malki and Reiter}{1996}]%
        {Malki:1996}
\bibfield{author}{\bibinfo{person}{Dalia Malki} {and} \bibinfo{person}{Michael
  Reiter}.} \bibinfo{year}{1996}\natexlab{}.
\newblock \showarticletitle{A High-Throughput Secure Reliable Multicast
  Protocol}. In \bibinfo{booktitle}{\emph{Proceedings of the 9th IEEE Workshop
  on Computer Security Foundations}} \emph{(\bibinfo{series}{CSFW '96})}.
  \bibinfo{publisher}{IEEE Computer Society}, \bibinfo{address}{USA},
  \bibinfo{pages}{9}.
\newblock
\showISBNx{0818675225}


\bibitem[\protect\citeauthoryear{Martin and Alvisi}{Martin and Alvisi}{2006}]%
        {Martin:2006}
\bibfield{author}{\bibinfo{person}{Jean-Philippe Martin} {and}
  \bibinfo{person}{Lorenzo Alvisi}.} \bibinfo{year}{2006}\natexlab{}.
\newblock \showarticletitle{Fast {Byzantine} Consensus}.
\newblock \bibinfo{journal}{\emph{IEEE Transactions on Dependable and Secure
  Computing}} \bibinfo{volume}{3}, \bibinfo{number}{3} (\bibinfo{date}{July}
  \bibinfo{year}{2006}), \bibinfo{pages}{202--215}.
\newblock
\showISSN{1545-5971}
\urldef\tempurl%
\url{https://doi.org/10.1109/TDSC.2006.35}
\showDOI{\tempurl}


\bibitem[\protect\citeauthoryear{Mazi{\`e}res and Shasha}{Mazi{\`e}res and
  Shasha}{2002}]%
        {Mazieres:2002}
\bibfield{author}{\bibinfo{person}{David Mazi{\`e}res} {and}
  \bibinfo{person}{Dennis Shasha}.} \bibinfo{year}{2002}\natexlab{}.
\newblock \showarticletitle{Building secure file systems out of {Byzantine}
  storage}. In \bibinfo{booktitle}{\emph{21st Symposium on Principles of
  Distributed Computing}} \emph{(\bibinfo{series}{PODC 2002})}.
  \bibinfo{publisher}{ACM}, \bibinfo{pages}{108--117}.
\newblock
\urldef\tempurl%
\url{https://doi.org/10.1145/571825.571840}
\showDOI{\tempurl}


\bibitem[\protect\citeauthoryear{Merkle}{Merkle}{1987}]%
        {Merkle:1987}
\bibfield{author}{\bibinfo{person}{Ralph~C. Merkle}.}
  \bibinfo{year}{1987}\natexlab{}.
\newblock \showarticletitle{A Digital Signature Based on a Conventional
  Encryption Function}. In \bibinfo{booktitle}{\emph{A Conference on the Theory
  and Applications of Cryptographic Techniques on Advances in Cryptology}}
  \emph{(\bibinfo{series}{CRYPTO 1987})}. \bibinfo{publisher}{Springer},
  \bibinfo{pages}{369--378}.
\newblock
\showISBNx{3540187960}
\urldef\tempurl%
\url{https://doi.org/10.1007/3-540-48184-2_32}
\showDOI{\tempurl}


\bibitem[\protect\citeauthoryear{Minsky, Trachtenberg, and Zippel}{Minsky
  et~al\mbox{.}}{2003}]%
        {Minsky:2003}
\bibfield{author}{\bibinfo{person}{Yaron Minsky}, \bibinfo{person}{Ari
  Trachtenberg}, {and} \bibinfo{person}{Richard Zippel}.}
  \bibinfo{year}{2003}\natexlab{}.
\newblock \showarticletitle{Set Reconciliation with Nearly Optimal
  Communication Complexity}.
\newblock \bibinfo{journal}{\emph{IEEE Transactions on Information Theory}}
  \bibinfo{volume}{49}, \bibinfo{number}{9} (\bibinfo{date}{Sept.}
  \bibinfo{year}{2003}), \bibinfo{pages}{2213--2218}.
\newblock
\showISSN{0018-9448}
\urldef\tempurl%
\url{https://doi.org/10.1109/TIT.2003.815784}
\showDOI{\tempurl}


\bibitem[\protect\citeauthoryear{Najafzadeh, Shapiro, and Eugster}{Najafzadeh
  et~al\mbox{.}}{2018}]%
        {Najafzadeh:2018bw}
\bibfield{author}{\bibinfo{person}{Mahsa Najafzadeh}, \bibinfo{person}{Marc
  Shapiro}, {and} \bibinfo{person}{Patrick Eugster}.}
  \bibinfo{year}{2018}\natexlab{}.
\newblock \showarticletitle{Co-Design and Verification of an Available File
  System}. In \bibinfo{booktitle}{\emph{19th International Conference on
  Verification, Model Checking, and Abstract Interpretation}}
  \emph{(\bibinfo{series}{VMCAI 2018})}. \bibinfo{publisher}{Springer},
  \bibinfo{pages}{358--381}.
\newblock
\urldef\tempurl%
\url{https://doi.org/10.1007/978-3-319-73721-8_17}
\showDOI{\tempurl}


\bibitem[\protect\citeauthoryear{Nakamoto}{Nakamoto}{2008}]%
        {Nakamoto:2008}
\bibfield{author}{\bibinfo{person}{Satoshi Nakamoto}.}
  \bibinfo{year}{2008}\natexlab{}.
\newblock \bibinfo{title}{{Bitcoin}: A Peer-to-Peer Electronic Cash System}.
\newblock
\newblock
\urldef\tempurl%
\url{https://bitcoin.org/bitcoin.pdf}
\showURL{%
\tempurl}


\bibitem[\protect\citeauthoryear{of~Standards and Technology}{of~Standards and
  Technology}{2002}]%
        {SHA2}
\bibfield{author}{\bibinfo{person}{National~Institute of Standards} {and}
  \bibinfo{person}{Technology}.} \bibinfo{year}{2002}\natexlab{}.
\newblock \bibinfo{title}{Secure Hash Standard (SHA) -- FIPS 180-2}.
\newblock
\newblock
\urldef\tempurl%
\url{https://csrc.nist.gov/publications/detail/fips/180/2/archive/2004-02-25}
\showURL{%
\tempurl}


\bibitem[\protect\citeauthoryear{Pearce and Hamano}{Pearce and Hamano}{2013}]%
        {GitHTTP}
\bibfield{author}{\bibinfo{person}{Shawn~O. Pearce} {and}
  \bibinfo{person}{Junio~C. Hamano}.} \bibinfo{year}{2013}\natexlab{}.
\newblock \bibinfo{title}{Git HTTP transfer protocols}.
\newblock
\newblock
\urldef\tempurl%
\url{https://www.git-scm.com/docs/http-protocol}
\showURL{%
\tempurl}


\bibitem[\protect\citeauthoryear{Pouwelse, Garbacki, Epema, and Sips}{Pouwelse
  et~al\mbox{.}}{2005}]%
        {Pouwelse:2005}
\bibfield{author}{\bibinfo{person}{Johan Pouwelse}, \bibinfo{person}{Paweł
  Garbacki}, \bibinfo{person}{Dick Epema}, {and} \bibinfo{person}{Henk Sips}.}
  \bibinfo{year}{2005}\natexlab{}.
\newblock \showarticletitle{The {BitTorrent} {P2P} File-Sharing System:
  Measurements and Analysis}. In \bibinfo{booktitle}{\emph{4th International
  Workshop on Peer-to-Peer Systems}} \emph{(\bibinfo{series}{IPTPS 2005})}.
  \bibinfo{pages}{205--216}.
\newblock
\urldef\tempurl%
\url{https://doi.org/10.1007/11558989_19}
\showDOI{\tempurl}


\bibitem[\protect\citeauthoryear{Pregui{\c c}a, Baquero, and Shapiro}{Pregui{\c
  c}a et~al\mbox{.}}{2018}]%
        {Preguica:2018gi}
\bibfield{author}{\bibinfo{person}{Nuno Pregui{\c c}a}, \bibinfo{person}{Carlos
  Baquero}, {and} \bibinfo{person}{Marc Shapiro}.}
  \bibinfo{year}{2018}\natexlab{}.
\newblock \showarticletitle{Conflict-Free Replicated Data Types ({CRDTs})}.
\newblock In \bibinfo{booktitle}{\emph{Encyclopedia of Big Data Technologies}}.
  \bibinfo{publisher}{Springer}.
\newblock
\showISBNx{978-3-319-63962-8}
\urldef\tempurl%
\url{https://doi.org/10.1007/978-3-319-63962-8_185-1}
\showDOI{\tempurl}


\bibitem[\protect\citeauthoryear{{Protocol Labs}}{{Protocol Labs}}{[n.d.]}]%
        {IPLD}
\bibfield{author}{\bibinfo{person}{{Protocol Labs}}.}
  \bibinfo{year}{[n.d.]}\natexlab{}.
\newblock \bibinfo{title}{{IPLD}}.
\newblock
\newblock
\urldef\tempurl%
\url{https://ipld.io/}
\showURL{%
\tempurl}


\bibitem[\protect\citeauthoryear{Robinson, Hand, Madsen, and McKelvey}{Robinson
  et~al\mbox{.}}{2018}]%
        {Robinson:2018}
\bibfield{author}{\bibinfo{person}{Danielle~C. Robinson},
  \bibinfo{person}{Joe~A. Hand}, \bibinfo{person}{Mathias~Buus Madsen}, {and}
  \bibinfo{person}{Karissa~R. McKelvey}.} \bibinfo{year}{2018}\natexlab{}.
\newblock \showarticletitle{The {Dat} Project, an open and decentralized
  research data tool}.
\newblock \bibinfo{journal}{\emph{Scientific Data}}  \bibinfo{volume}{5}
  (\bibinfo{date}{Oct.} \bibinfo{year}{2018}), \bibinfo{pages}{180221}.
\newblock
\showISSN{2052-4463}
\urldef\tempurl%
\url{https://doi.org/10.1038/sdata.2018.221}
\showDOI{\tempurl}


\bibitem[\protect\citeauthoryear{Schneider}{Schneider}{1990}]%
        {Schneider:1990}
\bibfield{author}{\bibinfo{person}{Fred~B. Schneider}.}
  \bibinfo{year}{1990}\natexlab{}.
\newblock \showarticletitle{Implementing Fault-Tolerant Services Using the
  State Machine Approach: A Tutorial}.
\newblock \bibinfo{journal}{\emph{Comput. Surveys}} \bibinfo{volume}{22},
  \bibinfo{number}{4} (\bibinfo{date}{Dec.} \bibinfo{year}{1990}),
  \bibinfo{pages}{299--319}.
\newblock
\showISSN{0360-0300}
\urldef\tempurl%
\url{https://doi.org/10.1145/98163.98167}
\showDOI{\tempurl}


\bibitem[\protect\citeauthoryear{Schwarz and Mattern}{Schwarz and
  Mattern}{1994}]%
        {Schwarz:1994}
\bibfield{author}{\bibinfo{person}{Reinhard Schwarz} {and}
  \bibinfo{person}{Friedemann Mattern}.} \bibinfo{year}{1994}\natexlab{}.
\newblock \showarticletitle{Detecting causal relationships in distributed
  computations: In search of the holy grail}.
\newblock \bibinfo{journal}{\emph{Distributed Computing}} \bibinfo{volume}{7},
  \bibinfo{number}{3} (\bibinfo{date}{March} \bibinfo{year}{1994}),
  \bibinfo{pages}{149--174}.
\newblock
\urldef\tempurl%
\url{https://doi.org/10.1007/BF02277859}
\showDOI{\tempurl}


\bibitem[\protect\citeauthoryear{Shapiro, Pregui{\c c}a, Baquero, and
  Zawirski}{Shapiro et~al\mbox{.}}{2011a}]%
        {Shapiro:2011wy}
\bibfield{author}{\bibinfo{person}{Marc Shapiro}, \bibinfo{person}{Nuno
  Pregui{\c c}a}, \bibinfo{person}{Carlos Baquero}, {and}
  \bibinfo{person}{Marek Zawirski}.} \bibinfo{year}{2011}\natexlab{a}.
\newblock \bibinfo{booktitle}{\emph{A comprehensive study of Convergent and
  Commutative Replicated Data Types}}.
\newblock \bibinfo{type}{{T}echnical {R}eport} 7506.
  \bibinfo{institution}{INRIA}.
\newblock
\urldef\tempurl%
\url{http://hal.inria.fr/inria-00555588/}
\showURL{%
\tempurl}


\bibitem[\protect\citeauthoryear{Shapiro, Pregui{\c c}a, Baquero, and
  Zawirski}{Shapiro et~al\mbox{.}}{2011b}]%
        {Shapiro:2011}
\bibfield{author}{\bibinfo{person}{Marc Shapiro}, \bibinfo{person}{Nuno
  Pregui{\c c}a}, \bibinfo{person}{Carlos Baquero}, {and}
  \bibinfo{person}{Marek Zawirski}.} \bibinfo{year}{2011}\natexlab{b}.
\newblock \showarticletitle{Conflict-Free Replicated Data Types}. In
  \bibinfo{booktitle}{\emph{13th International Symposium on Stabilization,
  Safety, and Security of Distributed Systems}} \emph{(\bibinfo{series}{SSS
  2011})}. \bibinfo{publisher}{Springer}, \bibinfo{pages}{386--400}.
\newblock
\urldef\tempurl%
\url{https://doi.org/10.1007/978-3-642-24550-3_29}
\showDOI{\tempurl}


\bibitem[\protect\citeauthoryear{Shoker, Yactine, and Baquero}{Shoker
  et~al\mbox{.}}{2017}]%
        {Shoker:2017}
\bibfield{author}{\bibinfo{person}{Ali Shoker}, \bibinfo{person}{Houssam
  Yactine}, {and} \bibinfo{person}{Carlos Baquero}.}
  \bibinfo{year}{2017}\natexlab{}.
\newblock \showarticletitle{As Secure as Possible Eventual Consistency: Work in
  Progress}. In \bibinfo{booktitle}{\emph{3rd International Workshop on
  Principles and Practice of Consistency for Distributed Data}}
  \emph{(\bibinfo{series}{PaPoC 2017})}. \bibinfo{publisher}{ACM}, Article
  \bibinfo{articleno}{5}.
\newblock
\urldef\tempurl%
\url{https://doi.org/10.1145/3064889.3064895}
\showDOI{\tempurl}


\bibitem[\protect\citeauthoryear{Singh, Castro, Druschel, and Rowstron}{Singh
  et~al\mbox{.}}{2004}]%
        {Singh:2004}
\bibfield{author}{\bibinfo{person}{Atul Singh}, \bibinfo{person}{Miguel
  Castro}, \bibinfo{person}{Peter Druschel}, {and} \bibinfo{person}{Antony
  Rowstron}.} \bibinfo{year}{2004}\natexlab{}.
\newblock \showarticletitle{Defending against Eclipse Attacks on Overlay
  Networks}. In \bibinfo{booktitle}{\emph{11th ACM SIGOPS European Workshop}}
  \emph{(\bibinfo{series}{EW 2011})}. \bibinfo{publisher}{ACM}.
\newblock
\urldef\tempurl%
\url{https://doi.org/10.1145/1133572.1133613}
\showDOI{\tempurl}


\bibitem[\protect\citeauthoryear{Singh, Fonseca, Kuznetsov, Rodrigues, and
  Maniatis}{Singh et~al\mbox{.}}{2009}]%
        {Singh:2009}
\bibfield{author}{\bibinfo{person}{Atul Singh}, \bibinfo{person}{Pedro
  Fonseca}, \bibinfo{person}{Petr Kuznetsov}, \bibinfo{person}{Rodrigo
  Rodrigues}, {and} \bibinfo{person}{Petros Maniatis}.}
  \bibinfo{year}{2009}\natexlab{}.
\newblock \showarticletitle{{Zeno}: Eventually Consistent {Byzantine}-Fault
  Tolerance}. In \bibinfo{booktitle}{\emph{6th USENIX Symposium on Networked
  Systems Design and Implementation}} \emph{(\bibinfo{series}{NSDI 2009})}.
  \bibinfo{publisher}{USENIX}, \bibinfo{pages}{169--184}.
\newblock


\bibitem[\protect\citeauthoryear{Skjegstad and Maseng}{Skjegstad and
  Maseng}{2011}]%
        {Skjegstad:2011}
\bibfield{author}{\bibinfo{person}{Magnus Skjegstad} {and}
  \bibinfo{person}{Torleiv Maseng}.} \bibinfo{year}{2011}\natexlab{}.
\newblock \showarticletitle{Low Complexity Set Reconciliation Using Bloom
  Filters}. In \bibinfo{booktitle}{\emph{7th ACM SIGACT/SIGMOBILE International
  Workshop on Foundations of Mobile Computing}} \emph{(\bibinfo{series}{FOMC
  2011})}. \bibinfo{publisher}{ACM}, \bibinfo{pages}{33--41}.
\newblock
\urldef\tempurl%
\url{https://doi.org/10.1145/1998476.1998483}
\showDOI{\tempurl}


\bibitem[\protect\citeauthoryear{Tan and Hamano}{Tan and Hamano}{2018}]%
        {GitSkipping}
\bibfield{author}{\bibinfo{person}{Jonathan~Tan Tan} {and}
  \bibinfo{person}{Junio~C. Hamano}.} \bibinfo{year}{2018}\natexlab{}.
\newblock \bibinfo{title}{negotiator/skipping: skip commits during fetch}.
\newblock \bibinfo{howpublished}{Git source code commit 42cc7485}.
\newblock
\urldef\tempurl%
\url{https://github.com/git/git/commit/42cc7485a2ec49ecc440c921d2eb0cae4da80549}
\showURL{%
\tempurl}


\bibitem[\protect\citeauthoryear{Tao, Shapiro, and Rancurel}{Tao
  et~al\mbox{.}}{2015}]%
        {Tao:2015gd}
\bibfield{author}{\bibinfo{person}{Vinh Tao}, \bibinfo{person}{Marc Shapiro},
  {and} \bibinfo{person}{Vianney Rancurel}.} \bibinfo{year}{2015}\natexlab{}.
\newblock \showarticletitle{Merging semantics for conflict updates in
  geo-distributed file systems}. In \bibinfo{booktitle}{\emph{8th ACM
  International Systems and Storage Conference}} \emph{(\bibinfo{series}{SYSTOR
  2015})}. \bibinfo{publisher}{ACM}.
\newblock
\urldef\tempurl%
\url{https://doi.org/10.1145/2757667.2757683}
\showDOI{\tempurl}


\bibitem[\protect\citeauthoryear{Tarr, Lavoie, Meyer, and Tschudin}{Tarr
  et~al\mbox{.}}{2019}]%
        {Tarr:2019}
\bibfield{author}{\bibinfo{person}{Dominic Tarr}, \bibinfo{person}{Erick
  Lavoie}, \bibinfo{person}{Aljoscha Meyer}, {and} \bibinfo{person}{Christian
  Tschudin}.} \bibinfo{year}{2019}\natexlab{}.
\newblock \showarticletitle{{Secure Scuttlebutt}: An Identity-Centric Protocol
  for Subjective and Decentralized Applications}. In
  \bibinfo{booktitle}{\emph{6th ACM Conference on Information-Centric
  Networking}} \emph{(\bibinfo{series}{ICN 2019})}.
\newblock
\urldef\tempurl%
\url{https://doi.org/10.1145/3357150.3357396}
\showDOI{\tempurl}


\bibitem[\protect\citeauthoryear{van~der Linde, Fouto, Leit{\~a}o, Pregui{\c
  c}a, Casti{\~n}eira, and Bieniusa}{van~der Linde et~al\mbox{.}}{2017}]%
        {vanderLinde:2017fu}
\bibfield{author}{\bibinfo{person}{Albert van~der Linde},
  \bibinfo{person}{Pedro Fouto}, \bibinfo{person}{Jo{\~a}o Leit{\~a}o},
  \bibinfo{person}{Nuno Pregui{\c c}a}, \bibinfo{person}{Santiago
  Casti{\~n}eira}, {and} \bibinfo{person}{Annette Bieniusa}.}
  \bibinfo{year}{2017}\natexlab{}.
\newblock \showarticletitle{{Legion}: Enriching Internet Services with
  Peer-to-Peer Interactions}. In \bibinfo{booktitle}{\emph{26th International
  Conference on World Wide Web}} \emph{(\bibinfo{series}{WWW 2017})}.
  \bibinfo{publisher}{ACM}, \bibinfo{pages}{283--292}.
\newblock
\urldef\tempurl%
\url{https://doi.org/10.1145/3038912.3052673}
\showDOI{\tempurl}


\bibitem[\protect\citeauthoryear{Van~Gundy and Chen}{Van~Gundy and
  Chen}{2012}]%
        {VanGundy:2012}
\bibfield{author}{\bibinfo{person}{Matthew~D. Van~Gundy} {and}
  \bibinfo{person}{Hao Chen}.} \bibinfo{year}{2012}\natexlab{}.
\newblock \bibinfo{title}{{OldBlue}: Causal Broadcast In A Mutually Suspicious
  Environment (Working Draft)}.
\newblock
\newblock
\urldef\tempurl%
\url{https://matt.singlethink.net/projects/mpotr/oldblue-draft.pdf}
\showURL{%
\tempurl}


\bibitem[\protect\citeauthoryear{van Hardenberg and Kleppmann}{van Hardenberg
  and Kleppmann}{2020}]%
        {vanHardenberg2020PushPin}
\bibfield{author}{\bibinfo{person}{Peter van Hardenberg} {and}
  \bibinfo{person}{Martin Kleppmann}.} \bibinfo{year}{2020}\natexlab{}.
\newblock \showarticletitle{{PushPin}: Towards Production-Quality Peer-to-Peer
  Collaboration}. In \bibinfo{booktitle}{\emph{7th Workshop on Principles and
  Practice of Consistency for Distributed Data}} \emph{(\bibinfo{series}{PaPoC
  2020})}. \bibinfo{publisher}{ACM}, Article \bibinfo{articleno}{10}.
\newblock
\urldef\tempurl%
\url{https://doi.org/10.1145/3380787.3393683}
\showDOI{\tempurl}


\bibitem[\protect\citeauthoryear{Vogels}{Vogels}{2009}]%
        {Vogels:2009ca}
\bibfield{author}{\bibinfo{person}{Werner Vogels}.}
  \bibinfo{year}{2009}\natexlab{}.
\newblock \showarticletitle{Eventually consistent}.
\newblock \bibinfo{journal}{\emph{Commun. ACM}} \bibinfo{volume}{52},
  \bibinfo{number}{1} (\bibinfo{year}{2009}), \bibinfo{pages}{40--44}.
\newblock
\urldef\tempurl%
\url{https://doi.org/10.1145/1435417.1435432}
\showDOI{\tempurl}


\bibitem[\protect\citeauthoryear{Vries}{Vries}{2020}]%
        {deVries:2020}
\bibfield{author}{\bibinfo{person}{Alex~de Vries}.}
  \bibinfo{year}{2020}\natexlab{}.
\newblock \showarticletitle{{Bitcoin}'s energy consumption is underestimated: A
  market dynamics approach}.
\newblock \bibinfo{journal}{\emph{Energy Research \& Social Science}}
  \bibinfo{volume}{70} (\bibinfo{date}{Dec.} \bibinfo{year}{2020}),
  \bibinfo{pages}{101721}.
\newblock
\showISSN{2214-6296}
\urldef\tempurl%
\url{https://doi.org/10.1016/j.erss.2020.101721}
\showDOI{\tempurl}


\bibitem[\protect\citeauthoryear{Weiss, Urso, and Molli}{Weiss
  et~al\mbox{.}}{2009}]%
        {Weiss:2009ht}
\bibfield{author}{\bibinfo{person}{St{\'e}phane Weiss}, \bibinfo{person}{Pascal
  Urso}, {and} \bibinfo{person}{Pascal Molli}.}
  \bibinfo{year}{2009}\natexlab{}.
\newblock \showarticletitle{{Logoot}: A Scalable Optimistic Replication
  Algorithm for Collaborative Editing on {P2P} Networks}. In
  \bibinfo{booktitle}{\emph{29th IEEE International Conference on Distributed
  Computing Systems}} \emph{(\bibinfo{series}{ICDCS 2009})}.
  \bibinfo{publisher}{IEEE}, \bibinfo{pages}{404--412}.
\newblock
\urldef\tempurl%
\url{https://doi.org/10.1109/ICDCS.2009.75}
\showDOI{\tempurl}


\bibitem[\protect\citeauthoryear{Zawirski, Pregui{\c c}a, Duarte, Bieniusa,
  Balegas, and Shapiro}{Zawirski et~al\mbox{.}}{2015}]%
        {Zawirski2015SwiftCloud}
\bibfield{author}{\bibinfo{person}{Marek Zawirski}, \bibinfo{person}{Nuno
  Pregui{\c c}a}, \bibinfo{person}{S{\'e}rgio Duarte}, \bibinfo{person}{Annette
  Bieniusa}, \bibinfo{person}{Valter Balegas}, {and} \bibinfo{person}{Marc
  Shapiro}.} \bibinfo{year}{2015}\natexlab{}.
\newblock \showarticletitle{Write Fast, Read in the Past: Causal Consistency
  for Client-side Applications}. In \bibinfo{booktitle}{\emph{16th Annual
  Middleware Conference}}. \bibinfo{publisher}{ACM/IFIP/USENIX},
  \bibinfo{pages}{75--87}.
\newblock
\urldef\tempurl%
\url{https://doi.org/10.1145/2814576.2814733}
\showDOI{\tempurl}


\bibitem[\protect\citeauthoryear{Zhao, Babi, Yang, Luo, Zhu, Yang, Luo, and
  Yang}{Zhao et~al\mbox{.}}{2016}]%
        {Zhao:2016}
\bibfield{author}{\bibinfo{person}{Wenbing Zhao}, \bibinfo{person}{Mamdouh
  Babi}, \bibinfo{person}{William Yang}, \bibinfo{person}{Xiong Luo},
  \bibinfo{person}{Yueqin Zhu}, \bibinfo{person}{Jack Yang},
  \bibinfo{person}{Chaomin Luo}, {and} \bibinfo{person}{Mary Yang}.}
  \bibinfo{year}{2016}\natexlab{}.
\newblock \showarticletitle{Byzantine Fault Tolerance for Collaborative Editing
  with Commutative Operations}. In \bibinfo{booktitle}{\emph{IEEE International
  Conference on Electro Information Technology}} \emph{(\bibinfo{series}{EIT
  2016})}. \bibinfo{publisher}{IEEE}, \bibinfo{pages}{246--251}.
\newblock
\urldef\tempurl%
\url{https://doi.org/10.1109/eit.2016.7535248}
\showDOI{\tempurl}


\end{thebibliography}
